\documentclass[a4paper]{amsart}
\usepackage{a4wide}
\usepackage{amsmath,amsfonts,amssymb,amsthm}

\theoremstyle{plain}
\newtheorem{theorem}{Theorem}[section]
\newtheorem{proposition}[theorem]{Proposition}
\newtheorem{corollary}[theorem]{Corollary}
\newtheorem{lemma}[theorem]{Lemma}

\theoremstyle{remark}
\newtheorem{remark}[theorem]{Remark}
\newtheorem{definition}[theorem]{Definition}

\newcommand{\R}{{\mathbb R}}

\newcommand{\kk}{{\bf k}}

\newcommand{\x}{{\bf x}}

\newcommand{\C}{{\mathbb C}}

\newcommand{\N}{{\mathbb N}}
\newcommand{\Z}{{\mathbb Z}}
\newcommand{\Hi}{\mathcal{H}}
\newcommand{\rueda}{\mathcal{I}}
\newcommand{\eps}{\varepsilon}
\newcommand{\Id}{\mathbf{1}}
\newcommand{\eu}{\mathrm{e}}
\newcommand{\iu}{\mathrm{i}}
\newcommand{\di}{\mathrm{d}}
\newcommand{\set}[1]{\left\{ #1 \right\}}
\newcommand{\norm}[1]{\left\| #1 \right\|}
\newcommand{\scal}[2]{\left\langle #1,\, #2 \right\rangle}
\newcommand{\bra}[1]{\left\langle #1 \right|}
\newcommand{\ket}[1]{\left| #1 \right\rangle}
\newcommand{\sub}[1]{_{\text{#1}}}
\DeclareMathOperator{\Ran}{Ran}
\DeclareMathOperator{\Tr}{Tr}

\title[Wannier functions in topological insulators: the 3D case]{On the construction of Wannier functions in\\ topological insulators: the 3D case}
\author{Horia D. Cornean \and Domenico Monaco}
\date{\today, arXiv v2} 

\usepackage{hyperref}

\begin{document}

\begin{abstract}
We investigate the possibility of constructing exponentially localized composite Wannier bases, or equivalently smooth periodic Bloch frames, for $3$-dimensional time-reversal symmetric topological insulators, both of bosonic and of fermionic type, so that the bases in question are also compatible with time-reversal symmetry. 

This problem is translated in the study (of independent interest) of homotopy classes of continuous, periodic, and time-reversal symmetric families of unitary matrices. We identify three $\Z_2$-valued complete invariants for these homotopy classes. When these invariants vanish, we provide an algorithm which constructs a ``multi-step logarithm" that is employed to continuously deform the given family into a constant one, identically equal to the identity matrix. This algorithm leads to a constructive procedure to produce the composite Wannier bases mentioned above.

\bigskip


\medskip

\end{abstract}

\maketitle

\tableofcontents


\section{Introduction} \label{sec:intro}

This paper is a follow-up to \cite{CorneanHerbstNenciu16} and \cite{CorneanMonacoTeufel17}. We study families of rank-$m$ projections $\set{P(\kk)}_{\kk \in \R^D}$, $P(\kk) = P(\kk)^2 = P(\kk)^*$, acting on some Hilbert space $\Hi$, which are subject to the following conditions:
\begin{enumerate}
 \item the map $P \colon \R^D \to \mathcal{B}(\Hi)$, $\kk \mapsto P(\kk)$, is smooth (at least of class $C^1$);
 \item the map $P \colon \R^D \to \mathcal{B}(\Hi)$, $\kk \mapsto P(\kk)$, is $\Z^D$-periodic, that is, $P(\kk) = P(\kk + \mathbf{n})$ for all $\mathbf{n} \in \Z^D$;
 \item the family of projections is \emph{time-reversal symmetric}, namely there exists an antiunitary operator $\theta \colon \Hi \to \Hi$, $\theta^2 = \pm \Id$, such that for all $\kk \in \R^D$
 \[ \theta \, P(\kk) \, \theta^{-1} = P(-\kk). \]
\end{enumerate}
We say that the time-reversal symmetry operator $\theta$ is of \emph{bosonic type} if $\theta^2 = + \Id$, while it is of \emph{fermionic type} if instead $\theta^2 = - \Id$. Notice that the latter case forces the rank $m$ of the projections, as well as the dimension of the Hilbert space $\Hi$ in case it is finite-dimensional, to be even.

Such families of projections arise in condensed matter physics from the Bloch-Floquet transform of a periodic, time-reversal symmetric Hamiltonian. We refer the reader to \cite{CorneanHerbstNenciu16, Kuchment93} for more details. Interesting examples of such physical systems come from the $2$- and $3$-dimensional \emph{time-reversal symmetric topological insulators} introduced by Fu, Kane and Mele \cite{FuKane06, FuKaneMele07}.

We address the possibility to construct a \emph{Bloch frame} for the family of projections $\set{P(\kk)}_{\kk \in \R^D}$, according to the following
\begin{definition} \label{def:BlochFrame}
A \emph{Bloch frame} for the family of projections $\set{P(\kk)}_{\kk \in \R^D}$ is a collection of maps $\Xi = \set{\xi_a}_{a=1}^{m}$, with $\xi_a \colon \R^D \to \Hi$, such that the vectors $\set{\xi_a(\kk)}_{a=1}^{m}$ give an orthonormal basis of $\Ran P(\kk)$ for all $\kk \in \R^D$, that is, $P(\kk) = \sum_{a=1}^{m} \ket{\xi_a(\kk)} \bra{\xi_a(\kk)}$.

A Bloch frame $\Xi$ is called
\begin{enumerate}
 \item \label{item:Bloch_a} \emph{continuous} if each map $\xi_a \colon \R^D \to \Hi$, $a \in \set{1, \ldots, m}$, is continuous;
 \item \label{item:Bloch_b} \emph{periodic} if each map $\xi_a \colon \R^D \to \Hi$, $a \in \set{1, \ldots, m}$, is $\Z^D$-periodic, that is, $\xi_a(\kk) = \xi_a(\kk + \mathbf{n})$ for all $\mathbf{n} \in \Z^D$;
 \item \label{item:Bloch_c} \emph{time-reversal symmetric} if for all $\kk \in \R^D$
 \[ \xi_b(-\kk) = \sum_{a=1}^{m} \left[\theta \, \xi_a(\kk)\right] \, \eps_{ab}, \quad b \in \set{1, \ldots, m} \]
 where $\eps = [\eps_{ab}]$ is the identity matrix $\Id$ in the bosonic case, and the standard symplectic matrix $J = \begin{pmatrix} 0 & \Id \\ - \Id & 0 \end{pmatrix}$ in the fermionic case.
\end{enumerate}
\end{definition}

Results concerning the \emph{existence} of such Bloch frames when $D \le 3$ were formulated in \cite{BrouderPanati07, Panati07} for what concerns the bosonic case, and in \cite{MonacoPanati15} for the fermionic case. These results establish that continuous and periodic Bloch frames for time-reversal symmetric families of projections always exist, but the proofs involve abstract methods from bundle theory. Computational physics \cite{SoluyanovVanderbilt12, WinklerSoluyanovTroyer16} motivated instead the need for more algorithmic proofs, which are also able to explicitly exhibiting these Bloch frames; moreover, the question of whether a time-reversal symmetry constraint can be imposed on the frames has been raised. After the pioneering works \cite{HelfferSjostrand89,Nenciu91} who gave the first constructive proofs of the existence of smooth and periodic Bloch frames for the case $m=1$ in any dimension, several proposal were put forward more recently \cite{FiorenzaMonacoPanati16_B, FiorenzaMonacoPanati16, CancesLevittPanatiStoltz17, CorneanHerbstNenciu16, CorneanMonacoTeufel17}, which emphasized how a \emph{topological obstruction} may arise in the fermionic case for $D=2$ and $D=3$. This topological obstruction is encoded in certain $\Z_2$-valued topological invariants, and is in compliance with the predictions of Fu, Kane and Mele in \cite{FuKane06, FuKaneMele07}.

\subsection{Main results}

In this paper we follow the approach already outlined in \cite{CorneanMonacoTeufel17}, where it was applied only to the case $D \le 2$, and extend it to $D=3$. Our strategy to construct a Bloch frame for the $3$-dimensional family of projections $\set{P(\kk)}_{\kk \in \R^3}$ relies on the solution of the same problem for its $2$-dimensional restriction on the plane where the first coordinate $k_1$ of $\kk$ is fixed to zero (compare Section~\ref{sec:Induction}). Notice that in the fermionic case, already the $2$-dimensional problem may be topologically obstructed by a non-zero $\Z_2$ invariant. If $\set{\Xi(0,k_2, k_3)}_{(k_2,k_3) \in \R^2}$ is a continuous, periodic and time-reversal symmetric Bloch frame for $\set{P(0,k_2,k_3)}_{(k_2,k_3) \in \R^2}$, then we can consider its image under the \emph{parallel transport} induced by the Berry connection along one full period in the $k_1$-direction. The parallel-transported frame will differ from $\Xi(0,k_2,k_3)$ by the action of a unitary $m \times m$ matrix $\alpha(k_2, k_3)$. If it is possible to continuously ``rotate'' this family of matrices back to the identity, then the combination of parallel transport and this ``rotation'' will define a frame $\Xi(\kk)$, $\kk \in \R^3$, which is continuous, $\Z^3$-periodic, and time-reversal symmetric (see Theorem~\ref{thm:Bloch} for a precise statement).

The topological obstruction to the existence of $\set{\Xi(\kk)}_{\kk \in \R^3}$ is thus encoded in the possibility to continuously deform $\alpha(k_2, k_3)$ into the identity, without breaking its symmetries (periodicity and a time-reversal symmetry constraint), which are induced by the ones of the frame $\set{\Xi(0,k_2,k_3)}_{(k_2, k_3) \in \R^2}$. Thus, we are naturally led to the  identification of the \emph{equivariant homotopy classes} of $2$-dimensional, continuous, periodic, and time-reversal symmetric families of unitary matrices. The first main result of this paper (Theorem~\ref{thm:2Dhomotopies} and Corollary~\ref{crl:2Dhomotopies}) establishes that, in the case of bosonic time-reversal symmetry, all such families are homotopically trivial (in the equivariant sense), while the equivariant homotopy classes of fermionic families are characterized by three $\Z_2$ invariants (leading to a total of $2^3 = 8$ different classes). These invariants can be obtained by considering the $\Z_2$ indices characterizing the equivariant homotopy classes of any three of the four $1$-dimensional restrictions to the lines $\set{k_2 = 0}$, $\set{k_2 = 1/2}$, $\set{k_3 = 0}$ and $\set{k_3=1/2}$ \cite{CorneanMonacoTeufel17}. If the time-reversal symmetry constraint is relaxed, instead, all families of continuous, periodic, and time-reversal symmetric unitary matrices are homotopic to the identity, regardless of the nature of the time-reversal symmetry operator; this in particular implies the existence of Bloch frames which are continuous and periodic for any family of periodic, time-reversal symmetric projections in $D=3$, and we recover the results of \cite{Panati07, MonacoPanati15}.

A constructive algorithm for the Bloch frames requires however to exhibit an explicit homotopy $\alpha_t$ between the given family of unitary matrices $\alpha$ and the identity. If $\alpha$ had a ``good'' logarithm, namely $\alpha(k_2,k_3) = \eu^{\iu h(k_2, k_3)}$ for a continuous family of self-adjoint matrices $\set{h(k_2,k_3)}_{(k_2,k_3) \in \R^2}$ satisfying the appropriate symmetries (periodicity and possibly time-reversal symmetry), then an homotopy would be simply given by setting $\alpha_t(k_2, k_3) := \eu^{\iu \, t \, h(k_2, k_3)}$. However, it is easily realized that not all families of unitary matrices admit a logarithm: the intuitive idea is that, as the parameters $(k_2,k_3)$ move, the eigenvalues of $\alpha(k_2,k_3)$ could come together and degenerate, thus preventing the possibility to choose a continuous branch cut in the resolvent of $\alpha(k_2,k_3)$ that can then be used to compute the logarithm. The second main result of this paper (Theorems~\ref{thm:generic} and \ref{thm:noTRS}) consists in showing, in a constructive manner, that any family $\alpha$ as above is \emph{arbitrarily close} to a family which \emph{does} admit a continuous and periodic logarithm, and that, when the $\Z_2$ invariants which characterize the equivariant homotopy class of $\alpha$ vanish, this logarithm can be chosen to be time-reversal symmetric as well. This in particular implies (Theorem~\ref{thm:m-s-log}) that $\alpha$ itself admits a \emph{two-step logarithm}, namely one can write
\[ \alpha(k_2, k_3) = \eu^{\iu h_2(k_2,k_3)/2} \eu^{\iu h_1(k_2,k_3)} \eu^{\iu h_2(k_2,k_3)/2} \]
for two continuous families $h_1, h_2$ of self-adjoint matrices satifying the prescribed symmetries. Again, this construction provides the required homotopy between $\alpha$ and $\Id$, from which one can build up the continuous, symmetric $3$-dimensional Bloch frame.

We would like to comment here on the comparison with the existing literature. To the best of our knowledge, the only other constructive approaches to the existence of Bloch frames which can be applied to $3$-dimensional time-reversal symmetric families of projections have been proposed by Fiorenza, Monaco and Panati \cite{FiorenzaMonacoPanati16_B, FiorenzaMonacoPanati16} and by Winkler, Soluyanov and Troyer \cite{WinklerSoluyanovTroyer16}. The main advantage of the approach presented here with respect to the one by Fiorenza, Monaco and Panati is that we are able to construct periodic frames also in the topologically obstructed case, explicitly breaking time-reversal symmetry in the fermionic setting. The technique developed by Winkler, Soluyanov and Troyer, on the contrary, focuses on the unobstructed case, and is not able to enforce the time-reversal symmetry constraint on the frame. Morever, it requires mapping the given system to a topologically trivial one through a path in some parameter space, which can be cumbersome in practice. Our method is more ``self-contained'' in this respect, in that it does not require any data other than the family of projections.

\subsection{Relation with Wannier functions and the Fu--Kane--Mele invariants}

In problems coming from condensed matter physics, modelled by a gapped, periodic, and time-reversal symmetric Hamiltonian $H$, the construction of Bloch frames implies the construction of a \emph{Wannier basis} for the occupied states of $H$, by transforming the frame back from the $\kk$-space representation to the position representation. The importance of \emph{Wannier functions} \cite{Wannier review} in computational solid state physics cannot be overstated: they are an essential tool, for example, for the computation of tight-binding effective dynamics, or to visualize chemical bonding and orbitals in solids.

The key feature of Wannier functions is their rate of decay at infinity. It is by now well-known that the existence of \emph{exponentially localized} Wannier functions is tantamount to the one of an analytic, periodic Bloch frame for the projection $\set{P(\kk)}_{\kk \in \R^D}$ over the occupied Bloch states, in the momentum representation (provided of course that $\kk \mapsto P(\kk)$ is real analytic); if $D \le 3$, time-reversal symmetry, of either bosonic or fermionic type, is then a sufficient condition to guarantee that such an exponentially localized Wannier basis exists \cite{BrouderPanati07}. On the other hand, the regularity of a Bloch frame can be always enhanced to analyticity once there exists a continuous one \cite{Panati07}, for example by convolution with an analytic (and even, in the time-reversal symmetric setting) kernel \cite{CorneanHerbstNenciu16, CorneanMonacoTeufel17}.

Imposing a further symmetry, like time-reversal, on the Wannier functions could potentially be of interest, so that for example the tight-binding description of the generating Hamiltonian preserves this symmetry. However, enforcing a new symmetry can in general lead to new topological obstructions to the localization of Wannier functions \cite{Monaco17}, as illustrated by the case of $2$- and $3$-dimensional fermionic time-reversal symmetric topological insulators. Our results in terms of Bloch frames then immediately translate in the possibility to \emph{construct exponentially localized and time-reversal symmetric Wannier functions in $3$-dimensional systems}
\begin{itemize}
 \item always, if the time-reversal symmetry is of bosonic nature;
 \item or provided four appropriately defined $\Z_2$ invariants vanish, in the fermionic case.
\end{itemize}

To conclude this Introduction, we will compare the four $\Z_2$-valued topological obstructions that we recover from our construction with the $\Z_2$ indices proposed by Fu, Kane and Mele in the context of $3$-dimensional time-reversal symmetric topological insulators \cite{FuKaneMele07}. As was mentioned above, our construction of a continuous and symmetric Bloch frame proceeds inductively on the dimension. As proved {\sl e.g.} in \cite{CorneanMonacoTeufel17}, it is always possible to find such a Bloch frame for a $1$-dimensional family of projections, hence we can construct the frame on the two lines $\set{k_1 = k_2 = 0}$ and $\set{k_1 = k_3 = 0}$: denote them by $\Xi(0,0,k_3)$ and $\Xi(0,k_2,0)$, respectively. By parallel transport of $\Xi(0,0,k_3)$ over one period in the $k_1$- and $k_2$-directions we obtain two $1$-dimensional family of unitary matrices which are continuous, periodic, and time-reversal symmetric; the same is true when we parallel transport $\Xi(0,k_2,0)$ in the $k_1$-direction. Each of these three families of matrices is characterized, up to equivariant homotopy, by one $\Z_2$ invariant: the results in \cite{GrafPorta13, FiorenzaMonacoPanati16, CorneanMonacoTeufel17} allow to identify these invariants with the $2$-dimensional Kane-Mele $\Z_2$ indices \cite{FuKane06} associated to the faces $\set{k_1=0}$, $\set{k_2=0}$ and $\set{k_3=0}$, and hence with the so-called weak invariants proposed by Fu, Kane and Mele for the $3$-dimensional topological insulators \cite{FuKaneMele07}. 

As was mentioned above, even if these three invariants vanish, there might still be topological obstruction to finding a $3$-dimensional continuous and symmetric Bloch frame. The fourth $\Z_2$ invariant to be computed is obtained as follows: provided for example that it is possible to construct the Bloch frame $\Xi(0,k_2,k_3)$ on the plane $\set{k_1=0}$ (starting from example from the already constructed frame $\Xi(0,0,k_3)$ moving in the $k_2$-direction; this requires the appropriate weak invariant computed above to vanish), one still needs to compare it with its parallel-transported version along one full period in the $k_1$-direction. The $2$-dimensional unitary matrix $\alpha(k_2,k_3)$ resulting from this comparison is characterized, up to equivariant homotopy, by its restriction at $\set{k_2=0}$, $\set{k_3=0}$ and $\set{k_2=1/2}$. Now, we have already computed the indices associated to the first two restrictions: these are the other two weak invariants. The index associated to the restriction at $\set{k_2=1/2}$ is a new, independent $\Z_2$ index. The so-called strong invariant of Fu, Kane and Mele can be obtained by multiplying in $\Z_2$ this index together with the one obtained at $\set{k_2=0}$.

\medskip

\noindent \textbf{Acknowledgments.} The authors would like to thank G.~Nenciu and G.~Panati for inspiring discussions. Financial support by Grant 4181-00042 of the Danish Council for Independent Research $|$ Natural Sciences and from the German Science Foundation (DFG) within the GRK 1838 ``Spectral theory and dynamics of quantum systems'' is gratefully acknowledged.


\section{Construction of Bloch frames: induction on the dimension} \label{sec:Induction}

This Section is devoted to the general problem of finding a continuous, $\Z^D$-periodic and time-reversal symmetric Bloch frame for a family of projections $\set{P(k_1,\kk)}_{(k_1,\kk) \in \R^D}$. We single out the first coordinate, as our construction will rely on the possibility to construct such a Bloch frame for the $d$-dimensional family $\set{P(0,\kk)}_{\kk \in \R^d}$, where $d = D-1$. In this sense, the construction of the Bloch frame is by induction on the dimension. 

We will relate the topological obstruction to the existence of the $D$-dimensional frame to an homotopy classification of certain $d$-dimensional families of unitary matrices, which we call \emph{matching matrices}. This classification will be investigated in the next Subsection. We will come back to Bloch frames in Section~\ref{sec:Frames}.

\subsection{Matching matrices, homotopy and logarithms} \label{sec:Matching}

Let $\eps$ denote either the identity matrix $\Id$ (in the bosonic case) or the standard symplectic matrix $J = \begin{pmatrix} 0 & \Id \\ - \Id & 0 \end{pmatrix}$ (in the fermionic case).

\begin{definition} \label{def:matching}
Let $\set{\alpha(\kk)}_{\kk \in \R^d} \subset M_m(\C)$ be a family of $m \times m$ matrices. We call it a \emph{family of matching matrices} if the following hold:
\begin{enumerate}
 \item the matrix $\alpha(\kk)$ is unitary for all $\kk \in \R^d$;
 \item the map $\alpha \colon \R^d \to U(m)$, $\kk \mapsto \alpha(\kk)$, is continuous;
 \item the map $\alpha \colon \R^d \to U(m)$, $\kk \mapsto \alpha(\kk)$, is $\Z^d$-periodic, that is, $\alpha(\kk) = \alpha(\kk + {\bf n})$ for all ${\bf n} \in \Z^d$;
 \item the family is \emph{time-reversal symmetric}, namely the relation
 \begin{equation} \label{alphaTRS}
 \eps \, \alpha(\kk) = \alpha(-\kk)^t \, \eps
 \end{equation}
 holds for all $\kk \in \R^d$.
\end{enumerate}
\end{definition}

\begin{definition} \label{def:equivariant}
Two families of matching matrices $\set{\alpha_0(\kk)}_{\kk \in \R^d}$ and $\set{\alpha_1(\kk)}_{\kk \in \R^d}$ are called \emph{equivariantly homotopic} if there exists a family of matrices $\set{\alpha_t(\kk)}_{t \in [0,1], \kk \in \R^d}$ which is continuous in $t \in [0,1]$, is a family of matching matrices for fixed $t \in [0,1]$, and is such that $\alpha_{t=0} \equiv \alpha_0$ and $\alpha_{t=1} \equiv \alpha_1$. Any family $\set{\alpha(\kk)}_{\kk \in \R^d}$ which is equivariantly homotopic to $\set{\alpha_1(\kk) \equiv \Id}_{\kk \in \R^d}$ is called \emph{equivariantly null-homotopic}.
\end{definition}

The class of null-homotopic matching matrices plays an important role in the construction of Bloch frames, as we will see in the next Subsection. In the following Theorem, we provide several characterizations for this class in arbitrary dimension $d$. In particular, the null-homotopic families of matching matrices are exactly the ones admitting a \emph{multi-step logarithm}, in the sense of Definition~\eqref{def:multilog} below. This result is well-known in the framework of $K$-theory of $C^*$-algebras (see {\sl e.g.} \cite[Exercise~4D]{Wegge-Olsen93}), but we give here a proof adapted to the present context for the reader's convenience.

\begin{definition} \label{def:multilog}
Let $\set{\alpha(\kk)}_{\kk \in \R^d}$ be a family of matching matrices. We say that $\alpha$ admits a \emph{multi-step logarithm} if there exist $M$ families of matrices $\set{h_\ell(\kk)}_{\kk \in \R^d}$ such that for all $\ell \in \set{1, \ldots, M}$
\begin{enumerate}
 \item \label{item:h_a} the matrix $h_\ell(\kk)$ is self-adjoint for all $\kk \in \R^d$;
 \item \label{item:h_b} the map $h_\ell \colon \R^d \to M_m(\C)$, $\kk \mapsto h_\ell(\kk)$, is continuous;
 \item \label{item:h_c} the map $h_\ell \colon \R^d \to M_m(\C)$, $\kk \mapsto h_\ell(\kk)$, is $\Z^d$-periodic, that is, $h_\ell(\kk) = h_\ell(\kk + {\bf n})$ for all ${\bf n} \in \Z^d$;
 \item \label{item:h_d} the family is \emph{time-reversal symmetric}, namely the relation
 \[ \eps \, h_\ell(\kk) = h_\ell(-\kk)^t \, \eps \]
 holds for all $\kk \in \R^d$;
 \item for all $\kk \in \R^d$
 \begin{equation} \label{multilog}
 \alpha(\kk) = \eu^{\iu \, h_M(\kk)/2} \cdots \eu^{\iu \, h_2(\kk)/2} \, \eu^{\iu h_1(\kk)} \, \eu^{\iu \, h_2(\kk)/2} \cdots \eu^{\iu \, h_M(\kk)/2}.
 \end{equation}
\end{enumerate}
\end{definition}

\begin{theorem} \label{TFAE}
Let $\set{\alpha(\kk)}_{\kk \in \R^d}$ be a family of matching matrices. Then the following are equivalent:
\begin{enumerate}
 \item \label{item:homotopy} the family is equivariantly null-homotopic;
 \item \label{item:logarithm} the family admits a multi-step logarithm;
 \item \label{item:beta} there exists a family of matrices $\set{\beta(k_1,\kk)}_{(k_1, \kk) \in \R^D}$, $D=d+1$, such that
 \begin{enumerate}
  \item \label{item:beta_a} the matrix $\beta(k_1,\kk)$ is unitary for all $(k_1,\kk) \in \R^D$;
  \item \label{item:beta_b} the map $\beta \colon \R^D \to U(m)$, $(k_1,\kk) \mapsto \beta(k_1,\kk)$, is continuous;
  \item \label{item:beta_c} for fixed $k_1 \in \R$, the map $\beta(k_1,\cdot) \colon \R^d \to U(m)$, $\kk \mapsto \beta(k_1,\kk)$, is $\Z^d$-periodic, that is, $\beta(k_1,\kk) = \beta(k_1,\kk + {\bf n})$ for all ${\bf n} \in \Z^d$;
  \item \label{item:beta_d} the relation
 \[ \beta(-k_1,-\kk) = \eps^{-1} \, \overline{\beta(k_1,\kk)} \, \eps \]
 holds for all $(k_1,\kk) \in \R^D$;
  \item for all $(k_1,\kk) \in \R^D$
  \begin{equation} \label{alpha_vs_beta}
  \alpha(\kk) = \beta(k_1,\kk) \, \beta(k_1+1,\kk)^{-1}.
  \end{equation}
 \end{enumerate}
\end{enumerate}
\end{theorem}
\begin{proof}
\noindent \fbox{\ref{item:homotopy} $\Longrightarrow$ \ref{item:logarithm}.} Let $\set{\alpha_t(\kk)}_{\kk \in \R^d}$ be a family of matching matrices depending continuously on $t \in [0,1]$ and such that $\alpha_0(\kk) = \Id$, $\alpha_1(\kk) = \alpha(\kk)$. Since $[0,1]$ is a compact interval and $\alpha_t$ is $\Z^d$-periodic, by uniform continuity there exists $\delta > 0$ such that
\begin{equation} \label{close}
\sup_{\kk \in \R^d} \norm{\alpha_s(\kk) - \alpha_t(\kk)} < 2 \quad \text{whenever} \quad |s-t|< \delta. 
\end{equation}
Let $M \in \N$ be such that $1/M < \delta$. Then in particular
\[ \sup_{\kk \in \R^d} \norm{\alpha_{1/M}(\kk) - \Id} < 2 \]
so that the Cayley transform \cite[Prop.~3.10]{CorneanMonacoTeufel17} provides a logarithm for $\alpha_{1/M}(\kk)$, {\sl i.e.} $\alpha_{1/M}(\kk) = \eu^{\iu h_M(\kk)}$, with $h_M$ satisfying \eqref{item:h_a}, \eqref{item:h_b}, \eqref{item:h_c}, and \eqref{item:h_d} in Definition~\ref{def:multilog}.

Using again \eqref{close} we have that
\[ \sup_{\kk \in \R^d} \norm{\eu^{-\iu \, h_M(\kk)/2} \, \alpha_{2/M}(\kk) \, \eu^{-\iu \, h_M(\kk)/2} - \Id} = \sup_{\kk \in \R^d} \norm{\alpha_{2/M}(\kk) - \alpha_{1/M}(\kk)} < 2 \]
so that by the same argument 
\[ \eu^{-\iu \, h_M(\kk)/2} \, \alpha_{2/M}(\kk) \, \eu^{-\iu \, h_M(\kk)/2} = \eu^{\iu h_{M-1}(\kk)}, \quad \text{or} \quad \alpha_{2/M}(\kk) = \eu^{\iu \, h_M(\kk)/2} \, \eu^{\iu \, h_{M-1}(\kk)} \, \eu^{\iu \, h_M(\kk)/2}. \]
Repeating the same line of reasoning $M$ times, we end up exactly with \eqref{multilog}.

\noindent \fbox{\ref{item:logarithm} $\Longrightarrow$ \ref{item:beta}.} For $k_1 \in [-1/2,1/2]$, set
\[ \beta(k_1,\kk) := \eu^{-\iu \, k_1 \, h_M(\kk)} \cdots \eu^{-\iu \, k_1 \, h_1(\kk)} \]
and extend this definition to $k_1 \in \R$ via
\begin{equation} \label{extend}
\beta(k_1,\kk) := \begin{cases} 
\alpha(\kk)^{-1} \, \beta(k_1-1,\kk) & \text{if } k_1 > 1/2,\\
\alpha(\kk) \, \beta(k_1+1,\kk) & \text{if } k_1 < -1/2.
\end{cases}
\end{equation}
We just need to show that this definition yields a continuous function of $k_1$. Indeed, using \eqref{multilog},
\begin{align*}
\beta(1/2+0,\kk) & = \alpha(\kk)^{-1} \, \beta(-1/2+0,\kk) = \alpha(\kk)^{-1} \,\eu^{\iu \, h_M(\kk)/2} \cdots \eu^{\iu \, h_1(\kk)/2} \\
& = \eu^{-\iu \, h_M(\kk)/2} \cdots \eu^{-\iu \, h_1(\kk)/2} = \beta(1/2-0,\kk).
\end{align*}

\noindent \fbox{\ref{item:beta} $\Longrightarrow$ \ref{item:homotopy}.} The required homotopy $\alpha_t$ between $\Id $ and $\alpha$ is provided by setting
\begin{equation} \label{ExplicitHomotopy}
\alpha_t(\kk) := \beta(-t/2,\kk) \, \beta(t/2,\kk)^{-1}, \quad t \in [0,1], \: \kk \in \R^d. \qedhere
\end{equation}
\end{proof}

\begin{remark}
If the original family of matching matrices is more than continuous (say, smooth or even analytic), the families $h_\ell$ (multi-step logarithms) constructed in Theorem~\ref{TFAE}(\ref{item:logarithm}) via the Cayley transform inherit the same regularity. However, in general we cannot expect more than continuity at $k_1 = p/2$, $p \in \Z$, for the family $\beta$ constructed in Theorem~\ref{TFAE}(\ref{item:beta}).

Indeed, assume that $M=2$ and $\beta(k_1,\kk)=\eu^{-\iu k_1 h_2(\kk)}\eu^{-\iu k_1 h_1(\kk)}$ if $k_1\in [-1/2,1/2]$. When we differentiate this expression at $k_1=1/2-0$ we get
$$\beta'(1/2-0,\kk)=-\iu h_2(\kk) \beta(1/2,\kk) - \iu \beta(1/2,\kk) h_1(\kk) $$
where `prime' denotes derivative with respect to the first variable $k_1$.

If $1/2<k_1<3/2$ we have according to \eqref{extend} that $\beta(k_1,\kk)=\alpha(\kk)^{-1}\beta(k_1-1,\kk)$, hence after differentiation at $k_1=1/2+0$ we get
$$\beta'(1/2+0,\kk)=\alpha(\kk)^{-1}\beta'(-1/2+0,\kk)=-\iu \alpha(\kk)^{-1} h_2(\kk) \beta(-1/2,\kk) -\iu \alpha(\kk)^{-1} \beta(-1/2,\kk) h_1(\kk).$$
Hence $\beta'$ is continuous at $1/2$ if and only if
$$h_2(\kk)\beta(1/2,\kk)= \alpha(\kk)^{-1}h_2(\kk)\beta(-1/2,\kk)$$
or
$$\alpha(\kk) h_2(\kk)=h_2(\kk) \alpha(\kk).$$

Since $\alpha(\kk) =\eu^{\iu h_2(\kk)/2}
\eu^{\iu h_1(\kk)}\eu^{\iu h_2(\kk)/2}$ it follows that $\eu^{\iu h_1(\kk)}$ must commute with $h_2$, hence $h_1$ commutes with $h_2$. This implies that the original family of matching matrices has a ``standard'' logarithm, which is known not to be true in general.
\end{remark}

\subsection{Consequence on Bloch frames} \label{sec:Frames}

Let $\set{P(k_1,\kk)}_{(k_1,\kk) \in \R^D}$, $D=d+1$, be a $D$-dimensional family of rank-$m$ projections which is smooth (at least $C^1$), $\Z^D$-periodic, and time-reversal symmetric (either of bosonic or fermionic type). Assume that the $d$-dimensional restriction $\set{P(0,\kk)}_{\kk \in \R^d}$ admits a smooth, periodic and time-reversal symmetric Bloch frame $\Xi(0,\kk) = \set{\xi_a(0,\kk)}_{1 \le a \le m}$. Let also $T_{\kk}(k_1,0)$ be the parallel transport unitary associated to the family $\set{P(k_1,\kk)}_{(k_1,\kk) \in \R^D}$ along the straight line between the points $(0,\kk)$ and $(k_1,\kk)$ (compare \cite[Sec.~3.1]{CorneanMonacoTeufel17}). Define then the matrix $\alpha(\kk)$ through the relation
\begin{equation} \label{matching}
[\alpha(\kk)]_{ab} := \scal{\xi_a(0,\kk)}{T_{\kk}(1,0) \xi_b(0,\kk)}.
\end{equation}
Notice indeed that $T_{\kk}(1,0) \Xi(0,\kk)$ and $\Xi(0,\kk)$ are orthonormal bases of the same vector space $\Ran P(0,\kk) = \Ran P(1,\kk)$. One easily verifies that $\set{\alpha(\kk)}_{\kk \in \R^d}$ is indeed a family of matching matrices in the sense of Definition~\ref{def:matching} (compare \cite[Prop.~3.2]{CorneanMonacoTeufel17}).

\begin{theorem} \label{thm:Bloch}
For the family of matching matrices $\set{\alpha(\kk)}_{\kk \in \R^d}$ defined by \eqref{matching}, any of the conditions in Theorem~\ref{TFAE} is in turn equivalent to the following:
\begin{enumerate}
 \setcounter{enumi}{3}
 \item there exists a Bloch frame $\Xi(k_1,\kk)$ for $\set{P(k_1,\kk)}_{(k_1,\kk) \in \R^D}$ which is continuous, $\Z^D$-periodic, and time-reversal symmetric.
\end{enumerate}
\end{theorem}
\begin{proof}
To show that condition \eqref{item:beta}  in Theorem~\ref{TFAE} is equivalent to the one in the present statement, it suffices to set
\[ \xi_a(k_1,\kk) := \sum_{b=1}^{m} \left[ T_{\kk}(k_1,0) \xi_b(0,\kk) \right] [\beta(k_1,\kk)]_{ba} \]
or equivalently
\[ [\beta(k_1,\kk)]_{ba} := \scal{T_{\kk}(k_1,0) \xi_b(0,\kk)}{\xi_a(k_1,\kk)} \]
(compare \cite[Prop.~3.3]{CorneanMonacoTeufel17}).
\end{proof}

\begin{remark} \label{rmk:noTRS}
From the proofs of Theorems~\ref{TFAE} and \ref{thm:Bloch}, one can see that if one drops the hypotheses of time-reversal symmetry (namely \eqref{item:h_d} in Definition~\ref{def:multilog} for the multi-step logarithm, \eqref{item:beta_d} in the statement of Theorem~\ref{TFAE} for the family $\set{\beta(k_1, \kk)}_{(k_1, \kk) \in \R^D}$, and \eqref{item:Bloch_c} in Definition~\ref{def:BlochFrame} for the Bloch frame), then one can still show that the following statements are equivalent:
\begin{enumerate}
 \item the family $\set{\alpha(\kk)}_{\kk \in \R^d}$ in \eqref{matching} is \emph{null-homotopic}, namely there exists a family of matrices $\set{\alpha_t(\kk)}_{t \in [0,1], \kk \in \R^d}$ which is continuous in $t \in [0,1]$, is continuous and $\Z^d$-periodic in $\kk$ for fixed $t \in [0,1]$, and is such that $\alpha_{t=0} \equiv \alpha$ and $\alpha_{t=1} \equiv \Id$ (no time-reversal symmetry is required for $t \in (0,1)$);
 \item the family $\set{\alpha(\kk)}_{\kk \in \R^d}$ admits a continuous and $\Z^d$-periodic multi-step logarithm;
 \item there exists a family $\set{\beta(k_1, \kk)}_{(k_1, \kk) \in \R^D}$ satisfying \eqref{item:beta_a}, \eqref{item:beta_b}, \eqref{item:beta_c} and \eqref{alpha_vs_beta} in the statement of Theorem~\ref{TFAE};
 \item there exists a Bloch frame $\Xi(k_1,\kk)$ for $\set{P(k_1,\kk)}_{(k_1,\kk) \in \R^D}$ which is continuous and $\Z^D$-periodic.
\end{enumerate}
\end{remark}

\subsection{\texorpdfstring{Homotopies of matching matrices in $d=1$ ($D=2$) and $d=2$ ($D=3$)}{Homotopies of matching matrices in d=1 (D=2) and d=2 (D=3)}}

As an illustration of the above concepts, we investigate the equivariant homotopy classes of families of matching matrices for $d \in \set{1,2}$ (corresponding to $D \in \set{2,3}$).

\begin{remark} \label{rmk:Properties}
We recall here a few relevant properties of matching matrices. We refer to \cite{CorneanHerbstNenciu16, CorneanMonacoTeufel17} for the proofs of the following statements.

\begin{itemize}
 \item If $\set{\alpha(\kk)}_{\kk \in \R^d}$ is a family of matching matrices, there exists a continuous, $\Z^d$-periodic, and even function $\phi \colon \R^d \to \R$ such that $\det \alpha(\kk) = \eu^{\iu \phi(\kk)}$ for all $\kk \in \R^d$. The family
 \[ \widetilde{\alpha}(\kk) := \eu^{-\iu \phi(\kk)/m} \alpha(\kk) \]
 is then again a family of matching matrices but takes values in the special unitary group $SU(m)$. At the level of frames, if $\Xi(0,\kk)$ is a continuous, $\Z^d$-periodic, and time-reversal symmetric Bloch frame for the $d$-dimensional family of projections $\set{P(0,\kk)}_{\kk \in \R^d}$, define
\[ \widetilde{\xi}_a(0,\kk) := \eu^{-\iu \phi(\kk)/2m} \, \xi_a(0, \kk), \quad a \in \set{1, \ldots, M}, \: \kk \in \R^d. \]
It is easily verified that the above gives a Bloch frame with the same properties. From \eqref{matching} one then reads that the new family of matching matrices $\widetilde{\alpha}(\kk)$ equals $\eu^{-\iu \phi(\kk)/m} \alpha(\kk)$ and hence has unit determinant.

 Notice that, if $m=1$, the above considerations immediately imply that $\alpha(\kk) = \eu^{\iu \phi(\kk)}$ is equivariantly null-homotopic in any dimension: the function $\phi$ exhibits its (``one-step'') logarithm. This result was obtained already in \cite{Nenciu91}.
 \item Consider a family of \emph{fermionic} time-reversal symmetric matching matrices. Let also $\kk_\sharp$ be a point such that $\kk_\sharp \equiv -\kk_\sharp \bmod \Z^d$ (so $\kk_\sharp = (k_1, \ldots, k_d)$ where each $k_j$ is of the form $p_j/2$ with $p_j \in \Z$). Then the matrix $\alpha(\kk_\sharp)$ has \emph{Kramers degenerate spectrum}, namely each of its eigenvalues has even degeneracy.
 \item Let $d=1$ and consider a family of \emph{fermionic} time-reversal symmetric matching matrices $\set{\alpha(k)}_{k \in \R}$. Then there is a well-defined $\Z_2$-valued index
 \begin{equation} \label{eqn:rueda}
 (-1)^{\rueda(\alpha)} = p(0) \, p(1/2), \quad \text{where} \quad p(k_\sharp) := \frac{\sqrt{\det \alpha(k_\sharp)}}{\mathrm{Pf} \left(\eps \alpha(k_\sharp) \right)}, \; k_\sharp \in \set{0,1/2},
 \end{equation}
 called the \emph{Graf--Porta index} (or \emph{GP-index} for short) \cite{GrafPorta13, CorneanMonacoTeufel17}. Notice that $p(k_\sharp)^2 = 1$.
\end{itemize}
\end{remark}

The next result gives a complete homotopy classification of $1$-dimensional families of matching matrices, as well as a description of their ``generic'' spectral properties.

\begin{theorem} \label{thm:1Dhomotopies}
Assume that $d=1$.
\begin{enumerate}
 \item \label{item:EquivariantB} Let $\eps = \Id$. Then any family of matching matrices $\set{\alpha(k)}_{k \in \R}$ is equivariantly null-homotopic.
 
 Moreover, there exists a sequence $\set{\alpha_n(k)}_{k \in \R}$, $n \in \N$, of families of matching matrices such that 
 \begin{itemize}
 \item $\sup_{k \in \R} \norm{\alpha_n(k) - \alpha(k)} \to 0$ as $n \to \infty$, and
 \item for any $n \in \N$ the spectrum of $\alpha_n(k)$ is completely non-degenerate for all $k \in \R$ .
 \end{itemize}
 \item \label{item:EquivariantF} Let $\eps = J$. Then two families of matching matrices are equivariantly homotopic if and only if their GP-indices are the same; in particular, a family of matching matrices $\set{\alpha(k)}_{k \in \R}$ is equivariantly null-homotopic if and only if $\rueda(\alpha) \in \Z_2$ vanishes. 
 
Moreover, there exists a sequence $\set{\alpha_n(k)}_{k \in \R}$, $n \in \N$, of families of matching matrices such that 
 \begin{itemize}
 \item $\sup_{k \in \R} \norm{\alpha_n(k) - \alpha(k)} \to 0$ as $n \to \infty$, and
 \item for any $n \in \N$ each eigenvalue of $\alpha_n(k)$ is exactly doubly degenerate for $k = p/2$, $p \in \Z$, while the spectrum of $\alpha_n(k)$ is completely non-degenerate for $k \neq p/2$, $p \in \Z$.
 \end{itemize}
\end{enumerate}
\end{theorem}
\begin{proof}
The statement on the equivariant homotopy classes in part~\ref{item:EquivariantB} is a consequence of \cite[Prop.~2.16]{CorneanHerbstNenciu16}, while the existence of the required approximants $\set{\alpha_n(k)}_{k \in \R}$ essentially follows from \cite[Lemma~2.18]{CorneanHerbstNenciu16}. 

Part \ref{item:EquivariantF} is instead the content of \cite[Thm.~5.12 and Prop.~5.9(2)]{CorneanMonacoTeufel17}.
\end{proof}

A first result of the present paper is the following generalization of the above statement, giving a characterization of the homotopy classes of families of matching matrices in dimension $d=2$.

\begin{theorem} \label{thm:2Dhomotopies}
Let $d=2$. Then two families of matching matrices $\set{\alpha_0(k_1, k_2)}_{(k_1,k_2) \in \R^2}$ and $\set{\alpha_1(k_1, k_2)}_{(k_1,k_2) \in \R^2}$ are equivariantly homotopic if and only if the four pairs of $1$-dimensional families of matching matrices given by
\begin{gather*}
\set{\alpha_0(0, k_2)}_{k_2 \in \R} \text{ and } \set{\alpha_1(0, k_2)}_{k_2 \in \R}, \\
\set{\alpha_0(1/2, k_2)}_{k_2 \in \R} \text{ and } \set{\alpha_1(1/2, k_2)}_{k_2 \in \R}, \\
\set{\alpha_0(k_1, 0)}_{k_1 \in \R} \text{ and } \set{\alpha_1(k_1, 0)}_{k_1 \in \R}, \\
\set{\alpha_0(k_1, 1/2)}_{k_1 \in \R} \text{ and } \set{\alpha_1(k_1, 1/2)}_{k_1 \in/ \R},
\end{gather*}
are equivariantly homotopic.
\end{theorem}
\begin{proof}
If the $2$-dimensional families are equivariantly homotopic, then the restriction of the homotopy to the appropriate line in $\R^2$ will give an equivariant homotopy between the above-mentioned $1$-dimensional families, so we must only prove the converse statement.

For $s \in [0,1]$, let
\begin{equation} \label{eqn:1dh}
\set{\alpha_s(0, k_2)}_{k_2 \in \R} , \quad \set{\alpha_s(1/2, k_2)}_{k_2 \in \R} , \quad \set{\alpha_s(k_1, 0)}_{k_1 \in \R} , \quad \set{\alpha_s(k_1, 1/2)}_{k_1 \in \R},
\end{equation}
be equivariant homotopies between the $1$-dimensional families mentioned in the statement. The goal is to construct an equivariant homotopy $\set{\alpha_s(k_1, k_2)}_{(k_1, k_2) \in \R^2}$, $s \in [0,1]$, which extends the above and modifies $\alpha_0$ into $\alpha_1$. We notice that it suffices to construct this homotopy on a fundamental domain
\[ (s, k_1, k_2) \in F := [0,1] \times [0,1/2] \times [-1/2, 1/2], \]
and then impose periodicity and time-reversal symmetry: Indeed, since the homotopies in \eqref{eqn:1dh} are already equivariant, this extension will be continuous on the whole $[0,1] \times \R^2$.

We consider the datum in \eqref{eqn:1dh} as defining a continuous map $\alpha$ on $\partial F$ with values in $U(m)$. Topologically, the boundary $\partial F$ is a $2$-sphere, and hence the map $\alpha$ determines an element of the second homotopy group $\pi_2(U(m))$ by considering its homotopy class. Since the latter homotopy group is trivial \cite[Chap.~8, Sect.~12]{Husemoller94}, $\alpha$ is null-homotopic, or equivalently it extends to the region $F$ which is enclosed by $\partial F$. This extension provides the required map leading to an equivariant homotopy between $\alpha_0$ and $\alpha_1$, as detailed above.
\end{proof}

Combining the above two results, we obtain the following

\begin{corollary} \label{crl:2Dhomotopies}
Assume that $d=2$.
\begin{enumerate}
 \item \label{item:2DEquivariantB} Let $\eps = \Id$. Then any family of matching matrices $\set{\alpha(\kk)}_{\kk \in \R^2}$ is equivariantly null-homotopic.
 \item \label{item:2DEquivariantF} Let $\eps = J$. Then two families of matching matrices $\set{\alpha_0(\kk)}_{\kk \in \R^2}$ and $\set{\alpha_1(\kk)}_{\kk \in \R^2}$ are equivariantly homotopic if and only if
 \begin{gather*}
 \rueda(\alpha_0(0, \cdot)) = \rueda(\alpha_1(0, \cdot)) , \quad \rueda(\alpha_0(1/2, \cdot)) = \rueda(\alpha_1(1/2, \cdot)), \\
 \rueda(\alpha_0(\cdot, 0)) = \rueda(\alpha_1(\cdot, 0)), \quad \text{and} \quad \rueda(\alpha_0(\cdot, 1/2)) = \rueda(\alpha_1(\cdot, 1/2)).
\end{gather*}
In particular, a family of matching matrices $\set{\alpha(\kk)}_{\kk \in \R}$ is equivariantly null-homotopic if and only if the four GP-indices $\rueda(\alpha(0, \cdot))$, $\rueda(\alpha(1/2, \cdot))$, $\rueda(\alpha(\cdot, 0))$ and $\rueda(\alpha(\cdot, 1/2))$ vanish in $\Z_2$. 
\end{enumerate}
\end{corollary}

\begin{remark}
Notice that the conditions listed in the above statement for a family of fermionic matching matrices to be equivariantly null-homotopic are not independent: if three of the above $1$-dimensional families have vanishing GP-indices, then so does the fourth. 

Indeed, let 
\begin{equation} \label{eqn:high-symmetry}
\kk_1 = (0,0), \quad \kk_2 = (1/2,0) , \quad \kk_3 = (1/2,1/2) \quad \text{and} \quad \kk_4=(1/2,0).
\end{equation}
Then
\begin{gather*}
(-1)^{\rueda(\alpha(\cdot,0))} = p(\kk_1) \, p(\kk_2), \quad (-1)^{\rueda(\alpha(1/2,\cdot))} = p(\kk_2) \, p(\kk_3), \\
(-1)^{\rueda(\alpha(\cdot,1/2))} = p(\kk_3) \, p(\kk_4), \quad (-1)^{\rueda(\alpha(0,\cdot))} = p(\kk_1) \, p(\kk_4),
\end{gather*}
where $p(\kk_i)$ is defined as in \eqref{eqn:rueda}. Since
\[ p(\kk_1) \, p(\kk_4) = p(\kk_1) \, p(\kk_2)^2 \, p(\kk_3)^2 \, p(\kk_4), \]
it follows that
\[ (-1)^{\rueda(\alpha(0,k_2))} = (-1)^{\rueda(\alpha(k_1,0))} \, (-1)^{\rueda(\alpha(1/2,k_2))} \, (-1)^{\rueda(\alpha(k_1,1/2))}. \]
There are then only three independent GP-indices among the ones of $\alpha(k_1,0)$, $\alpha(0,k_2)$, $\alpha(k_1,1/2)$ and $\alpha(1/2,k_2)$.
\end{remark}

\begin{remark} \label{rmk:PeriodicFrames}
If we drop the assumptions of time-reversal symmetry as in Remark~\ref{rmk:noTRS}, then one can show that when $d \in \set{1,2}$ \emph{any} family of matching matrices is null-homotopic, regardless of whether it is of bosonic or fermionic nature. This is done in \cite{CorneanHerbstNenciu16, CorneanMonacoTeufel17} for $d=1$, and in Theorem~\ref{thm:m-s-log} below for $d=2$. In particular, continuous and $\Z^D$-periodic Bloch frames for $D$-dimensional continuous, periodic, and time-reversal symmetric families of projections always exist (and can be explicitly constructed) when $D \le 3$.
\end{remark}

\subsection{Summary}

To summarize the above considerations, we see that in order to construct a Bloch frame for a $D$-dimensional family of projections, we need the following ingredients:
\begin{itemize}
 \item a Bloch frame for the $d$-dimensional restriction of the family on the hyperplane $\set{k_1 = 0} \subset \R^D$ ($d=D-1$);
 \item a multi-step logarithm (in the sense of Definition~\ref{def:multilog}) for the corresponding family of matching matrices, defined via \eqref{matching}.
\end{itemize}
Theorem~\ref{thm:1Dhomotopies} shows that the second condition is in general topologically obstructed. Nonetheless, in the unobstructed case it is possible to provide an explicit algorithm to produce the required multi-step logarithm. For example, the case $d=1$ (corresponding to $D=2$) was analysed thoroughly in \cite{CorneanHerbstNenciu16} for the bosonic case and in \cite{CorneanMonacoTeufel17} for the fermionic case.

In the following we study the case $d=2$ (and correspondingly $D=3$). The construction of a (multi-step) logarithm for a family of unitary matrices requires in general its approximation by matrices that lift any spectral degeneracy which is not dictated by symmetry (as is the case for Kramers degeneracy). As an illustrative example, we treat the case of families of $2 \times 2$ matching matrices in the next Section, both for $d=1$ and $d=2$. This already displays all the issues to be faced in the general setting. The case of matching matrices of arbitrary rank in $d=2$ will be addressed in Section~\ref{sec:2d}.

\begin{remark}
In what follows, we will implicitly assume (unless otherwise stated) that families of matching matrices, as well as their families of multi-step logarithms, are \emph{smooth} (at least of class $C^1$). The general case of \emph{continuous} families of matching matrices can be recovered by first taking a convolution with a smooth, even kernel (compare \cite[Lemma~2.3]{CorneanHerbstNenciu16} and \cite[Lemma~A.2]{CorneanMonacoTeufel17}).
\end{remark}


\section{Construction of the multi-step logarithm: rank 2} \label{sec:m=2}

Throughout this Section, $\set{\alpha(\kk)}_{\kk \in \R^d}$, $d \le 2$, is a family of $2 \times 2$ matching matrices. Without loss of generality, we moreover assume that $\alpha(\kk) \in SU(2)$ (compare Remark~\ref{rmk:Properties}). 

Recall that any smooth, $\Z^d$-periodic map $\R^d \ni \kk \mapsto \alpha(\kk) \in SU(2)$ can be represented as
\begin{equation} \label{SU(2)}
\alpha(\kk) = m(\kk)\Id + \iu \sum_{j=1}^{3} F_j(\kk) \, \sigma_j,
\end{equation}
where $m$ and $F_j$, $j \in \set{1, 2, 3}$, are smooth, $\Z^d$-periodic and real-valued functions satisfying
\[ m(\kk)^2 +\sum_{j=1}^{3} F_j(\kk)^2 = 1, \]
where $\set{\sigma_1, \sigma_2, \sigma_3}$ are the Pauli matrices.

The general strategy to find a multi-step logarithm goes as follows: One first constructs an approximation $\widetilde{\alpha}$ of $\alpha$ which has non-degenerate spectrum and lies sufficiently close to $\alpha$ in the norm topology. Due to the non-degeneracy of the spectrum of $\widetilde{\alpha}$, it is possible to find a branch cut for the logarithm which always lies in its resolvent set, and hence $\widetilde{\alpha}$ has a ``good'' logarithm in the sense of Definition~\ref{def:multilog}, namely $\widetilde{\alpha}(\kk) = \eu^{\iu \, h_2(\kk)}$. The fact that $\widetilde{\alpha}$ is close to $\alpha$ implies that $\eu^{-\iu \, h_2(\kk)/2} \, \alpha(\kk) \, \eu^{-\iu \, h_2(\kk)/2}$ is close to $\Id$ uniformly in $\kk$, and hence it admits a logarithm $h_1(\kk)$. The combination of $h_1$ and $h_2$ as in \eqref{multilog} gives the desired multi-step logarithm of $\alpha$.

From the above discussion it becomes clear that we need to lift the spectral degeneracies of $\alpha$.  One can show that the spectrum of $\alpha(\kk)$ as in \eqref{SU(2)} is then given by
\[ \sigma(\alpha(\kk)) = \set{m(\kk) \pm \iu \norm{{\bf F}(\kk)}}, \quad {\bf F}(\kk) := \left[ F_1(\kk), F_2(\kk), F_3(\kk) \right], \]
so that points $\kk$ for which $\alpha(\kk)$ has degenerate spectrum coincide with zeroes of the vector field ${\bf F}(\kk)$. Thus, the construction of a multi-step logarithm for $\alpha$ will be achieved by perturbing the vector field ${\bf F}$ so that it avoids zero.

\subsection{\texorpdfstring{The case $d=1$ ($D=2$)}{The case d=1 (D=2)}}

We start by considering the $1$-dimensional bosonic case. Time-reversal symmetry of $\alpha(k)\in SU(2)$ reads then $\alpha(k)=\alpha(-k)^t$. One can check that $\sigma_1 = \sigma_1^t$ and $\sigma_3 = \sigma_3^t$, while $\sigma_2 =- \sigma_2^t$. Thus the time-reversal symmetry of $\alpha(k)$ implies
\[ m(-k) = m(k), \quad F_1(-k)= F_1(k),\quad F_3(-k)= F_3(k), \quad \text{and}\quad F_2(-k)= -F_2(k) \]
for the functions $m$ and $F_j$ appearing in \eqref{SU(2)}.

We see that $F_2(0)=0$. We want to slightly perturb $\alpha$ so that the perturbed matrix $\alpha_s$ has non-degenerate spectrum and is still a family of matching matrices. Because the map $k \mapsto [F_1(k), F_3(k)]$ traces a closed smooth curve in $\R^2$, the origin is not an interior point of this curve by Sard's lemma. 
Thus given $s>0$ we may find a vector ${\bf v}^{(s)}=[v_1^{(s)},v_3^{(s)}]\in \R^2$ with $\norm{\bf v} =s$ such that 
$$\inf_{k\in \R} \norm{ \left[ F_1(k)+v_1^{(s)},F_3(k)+v_3^{(s)} \right] } >0.$$

 Define $F_{1,s}(k)=F_1(k)+v_1^{(s)}$, $F_{2,s}(k)=F_2(k)$, $F_{3,s}(k)=F_3(k)+v_3^{(s)}$, and 
$${\bf F}_s(k)=[F_{1,s}(k),F_{2,s}(k),F_{2,s}(k)].$$ 
If $s$ is small enough we have $m(k)^2 +\norm{{\bf F}_s(k)}^2\geq 1/2$ for all $k$ and we can define
\begin{equation}\label{eqn:1Dalpha_eps}
\alpha_s(k):=\frac{1}{\sqrt{m(k)^2 +\norm{{\bf F}_s(k)}^2 }}\left (m(k)\Id + \iu \sum_{j=1}^{3} F_{j,s}(k) \, \sigma_j\right ).
\end{equation}
We see that ${\bf F}_s(k)$ can never be zero, hence the matrix $\alpha_s$ has non-degenerate spectrum. Moreover, it converges to $\alpha$ as $s \to 0$, it is smooth, periodic, and time-reversal symmetric. Thus $\alpha$ admits a two-step logarithm. 

\bigskip

We now come to the fermionic case. This case is more involved, first of all because not all fermionic families of matching matrices are equivariantly null-homotopic, and hence admit a multi-step logarithm (Theorem~\ref{TFAE}).

Homotopy classes of fermionic matching matrices in this dimension are described by Theorem~\ref{thm:1Dhomotopies}\eqref{item:EquivariantF}, and in particular the equivariant homotopy class of $\alpha$ is characterized by its GP-index $\rueda(\alpha)$. In the case of $2 \times 2$ matrices of unit determinant, this index is easily computable. Indeed, one just needs to look at the Kramers degenerate spectrum of $\alpha(0)$ and $\alpha(1/2)$. Combining Kramers degeneracy with the constraint $\det \alpha(k)=1$ we obtain that the degenerate eigenvalues of $\alpha(0)$ and $\alpha(1/2)$ must be either $1$ or $-1$. Then one easily checks that $\rueda(\alpha)=0$ if and only if the spectra of $\alpha(0)$ and $\alpha(1/2)$ coincide.

In the following, we construct a multi-step logarithm (with moreover a number of steps $M \le 3$) for $\alpha$, under the assumption that both $\alpha(0)$ and $\alpha(1/2)$ have the same spectrum. Let us stress once again that the \emph{existence} of the multi-step logarithm is guaranteed by the general argument in Theorem~\ref{TFAE} (together with the characterization of the equivariant homotopy classes in Theorem~\ref{thm:1Dhomotopies}), but we look for a \emph{constructive algorithm} to produce it.

Since we assumed that the two matrices have the same spectrum, then by the previous considerations $\alpha(0)=\alpha(1/2)=\pm \Id$. Let us show that the case $\alpha(0)=\alpha(1/2)=- \Id$ can be reduced to the other one. Indeed, in this case we define $h_1(k):=\pi\Id$ and introduce the family
\[ \alpha_1(k) = \eu^{-\iu \, h_1(k)/2} \, \alpha(k) \, \eu^{-\iu \, h_1(k)/2}. \]
We see that $\alpha_1(k)$ remains a family of fermionic $2 \times 2$ matching matrices, and moreover $\alpha_1(0)= \alpha_1(1/2)=\Id$. Hence we can assume without loss of generality that $\alpha(0)=\alpha(1/2)= \Id$.

Looking at the representation \eqref{SU(2)} for $\alpha(k)$, fermionic time-reversal symmetry implies this time
\[ m(-k) = m(k) \quad \text{and} \quad F_j(-k)= -F_j(k), \: j \in \set{1,2,3}. \]
In particular, $m(0) = m(\pm 1/2) = 1$ and $F_j(0) = F_j(\pm 1/2) = 0$, and hence there exists a closed interval $I\subset (0,1/2)$ such that 
\[ m(k) \geq 0, \quad  k \in [0,1/2] \setminus I. \]

The vector field ${\bf F}(k)=[F_1(k),F_2(k),F_3(k)]$ is smooth, $\Z$-periodic, and odd. The restriction of ${\bf F}(k)$ to $I$ defines a smooth curve in $\R^3$. The origin in $\R^3$ is not an interior point of the range of ${\bf F}$ by Sard's lemma, hence given $s>0$ one can find a vector ${\bf v}^{(s)}$ such that $\norm{{\bf v}^{(s)}}=s$ and $\norm{{\bf F}(k)+{\bf v}^{(s)}}>0$ on $I$. Because $I$ is compact, we can find $c_s>0$ such that
\[ \inf_{k\in I} \norm{ {\bf F}(k)+{\bf v}^{(s)} } \geq c_s > 0. \]
Let $0\leq \chi \leq 1$ be a smooth function which equals $1$ on $I$ and has support in $(0,1/2)$. Define 
\[ {\bf F}_s(k):={\bf F}(k)+ (\chi(k)-\chi(-k)) {\bf v}^{(s)}, \quad k\in [-1/2,1/2], \]
and extend it to $\R$ by periodicity:
\[ {\bf F}_s(k)={\bf F}(k)+ \sum_{n\in\Z}(\chi(k-n)-\chi(-k-n)) {\bf v}^{(s)},\quad k\in \R. \]
Clearly, ${\bf F}_s(-k) = -{\bf F}_s(k)$. Also, $m(k)^2 + \norm{{\bf F}_s(k)}^2 \geq 1/2$ if $s$ is small enough, and we can define $\alpha_s$ as in \eqref{eqn:1Dalpha_eps}. Denote by
\[ m_s(k) := \frac{m(k)}{\sqrt{m(k)^2+\norm{{\bf F}_s(k)}^2}} . \]
We see that $m_s$ can never be $-1$ on $[-1/2,1/2]$ if $s$ is small enough: if $k\not \in I\cup (-I)$ then $m$ is positive, while if $k \in I\cup (-I)$ then $\norm{{\bf F}_s(k)}$ is bounded from below by a positive number. This implies that $-1$ is never in the spectrum of $\alpha_s$, and we can define
\[ h_2(k) := \frac{\arccos m_s(k)}{\norm{{\bf F}_s(k)}} \sum_{j=1}^{3} F_{s,j}(k) \, \sigma_j, \quad \Tr h_2(k) = 0, \quad \alpha_s(k)=\eu^{\iu h_2(k)}. \]
We observe that $h_2$ is smooth (because $\frac{\arccos(x)}{\sqrt{1-x^2}}$ is $C^\infty$ on $(-1,1]$) and obeys the properties listed in Definition~\ref{def:multilog}. Since $\alpha_s(k)$ converges in norm to $\alpha(k)$, we conclude by the considerations at the beginning of this Section that $\alpha$ admits a multi-step logarithm.  

\subsection{\texorpdfstring{The case $d=2$ ($D=3$)}{The case d=2 (D=3)}}

We now move to the $2$-dimensional case, and as before we start by considering a bosonic family of matching matrices in the form \eqref{SU(2)}. Given $s_1>0$ we can find a vector ${\bf v}^{(s_1)}= [v_1^{(s_1)}, v_3^{(s_1)}] \in \R^2$ with $\norm{{\bf v}^{(s_1)}}=s_1$ such that 
$$\inf_{k_2\in \R} \norm{ \left[F_1(0,k_2)+v_1^{(s_1)},F_3(0,k_2)+v_3^{(s_1)} \right]} >0.$$
Let $0\leq \eta \leq 1$ be a $C_0^\infty(\R)$ \emph{even} function that equals $1$ on $[-1/10,1/10]$, and has support in $(-1/5,1/5)$. Extend $\eta$ to $\R$ by $\Z$-periodicity:
\begin{equation} \label{eqn:etap}
\eta\sub{p}(x) := \sum_{n\in \Z} \eta (x-n).
\end{equation}
We have $\eta\sub{p}(x)=\eta\sub{p}(-x)$ and $\eta\sub{p}(x+1)=\eta\sub{p}(x)$. Define 
$${\bf F}_{s_1}(\kk):= \left[ F_1(\kk)+\eta_p(k_1) \, v_1^{(s_1)} , F_2(\kk) , F_3(\kk)+\eta_p(k_1)\, v_3^{(s_1)} \right].$$
There exists a small strip of width $\delta_1<1/10$ around the line $k_1=0$ such that $\norm{{\bf F}_{s_1}(\kk)}$ is bounded from below by a positive constant if $\kk$ belongs to this strip. 

Now we perturb ${\bf F}_{s_1}(\kk)$ (only its first and third components) around the line $k_1=1/2$, so that the new ${\bf F}_{s_2}(\kk)$ will be away from zero on narrow strips around both $k_1=0$ and $k_1=1/2$, and then make  it periodic and symmetric as before. Finally, we perturb again around a narrow horizontal strip around $k_2=1/2$ and get ${\bf F}_{s_3}(\kk)$, where we have to make sure that $s_3$ is small enough so that we do not destroy the non-vanishing property on the ``vertical'' strips. By periodicity, the same property will hold near $k_2=-1/2$. 

In order to simplify notation, we may assume that the original ${\bf F}(\kk)$ is non-zero near the boundary of $\Omega'=(0,1/2)\times (-1/2,1/2)$. Let $K$ be a compact included in $\Omega'$ such that ${\bf F}(\kk)$ is away from zero on the compact set $\overline{\Omega'\setminus K}$. Let $0\leq \chi\leq 1$ be a smooth function which equals $1$ on $K$ and $0$ outside $\Omega'$. Given $s>0$ we can find a $3$-dimensional vector ${\bf v}^{(s)} \in \R^3$ with $\norm{{\bf v}^{(s)}} =s$ such that 
$$\inf_{\kk\in K}\norm{{\bf F}(\kk)+{\bf v}^{(s)}} >0.$$
Then if $s$ is small enough (in order not to destroy the non-vanishing property of ${\bf F}$ near the boundary of $\Omega'$), the function 
$${\bf F}_s(\kk)={\bf F}(\kk)+\sum_{{\bf n}\in \Z^2} \left \{(\chi(\kk-{\bf n})+\chi(-\kk-{\bf n})) \left[v_1^{(s)},0,v_3^{(s)} \right] + (\chi(\kk-{\bf n})-\chi(-\kk-{\bf n})) \left[0,v_2^{(s)},0\right]\right \}$$
will be periodic, never zero, and obeying the necessary symmetry. Hence we can construct a multi-step logarithm as before. 

\bigskip

We switch now to the case of a $2$-dimensional family of fermionic $2 \times 2$ matching matrices. In this case, by virtue of Corollary~\ref{crl:2Dhomotopies} we have to look at the GP-indices of the four families $\set{\alpha(0,k_2)}_{k_2 \in \R}$, $\set{\alpha(1/2,k_2)}_{k_2 \in \R}$, $\set{\alpha(k_1,0)}_{k_1 \in \R}$, and $\set{\alpha(k_1,1/2)}_{k_1 \in \R}$ in order to ensure that the family $\alpha$ is null-homotopic and hence admits a multi-step logarithm. Notice that these restrictions are indeed 1-dimensional families of matching matrices, hence their GP-indices are well-defined. When all these indices vanish, we are able to construct a multi-step logarithm. Similarly to the $1$-dimensional case, the vanishing of these indices is equivalent to the fact that the matrices $\alpha(0,0)$, $\alpha(1/2,0)$,  $\alpha(0,1/2)$ and $\alpha(1/2,1/2)$ have the same spectrum. Moreover, because of the Kramers degeneracy of their spectrum, the four matrices listed above are all simultaneously equal to $\pm \Id$. By the same reduction argument as in $d=1$, we can consider that they equal $\Id$. 

Consider the family $\gamma_1(k_1):=\alpha(k_1,0)$, with $k_1\in \R$. By the argument provided in the previous Subsection for the fermionic $1$-dimensional case, we can find two families $g_{11}(k_1)$ and $g_{12}(k_1)$ obeying the properties listed in Definition~\ref{def:multilog} and 
\[ \eu^{-\iu \, g_{12}(k_1)/2} \eu^{-\iu \, g_{11}(k_1)/2} \gamma_1(k_1) \eu^{-\iu \, g_{11}(k_1)/2} \eu^{-\iu \, g_{12}(k_1)/2} = \Id. \]
Moreover, the $g$'s are traceless and equal zero at $0$ and $\pm 1/2$. 

Let $\eta\sub{p}(x)$ be the function in \eqref{eqn:etap}. The families 
\[ \widetilde{g}_{11}(k_1,k_2) := g_{11}(k_1)\,  \eta\sub{p}(k_2) \quad \text{and} \quad \widetilde{g}_{12}(k_1,k_2) := g_{12}(k_1)\, \eta\sub{p}(k_2) \]
obey the properties of a multi-step logarithm. Define
\[ \alpha_1(\kk) := \eu^{-\iu \, \widetilde{g}_{12}(\kk)/2} \eu^{-\iu \, \widetilde{g}_{11}(\kk)/2} \alpha(\kk) \eu^{-\iu \, \widetilde{g}_{11}(\kk)/2} \eu^{-\iu \, \widetilde{g}_{12}(\kk)/2}. \]
This new family will again be a family of matching matrices, and in addition $\alpha_1(k_1,0) = \Id$ for all $k_1 \in \R$. Moreover, $\alpha_1(\pm 1/2,0)=\alpha_1(0,\pm 1/2) = \alpha_1(\pm 1/2,\pm 1/2) = \Id$. 

Define $\gamma_2(k_1) := \alpha_1(k_1,\pm 1/2)$. Repeating the previous argument we can construct two smooth families ${g}_{21}(k_1)$ and ${g}_{22}(k_1)$ as in Definition~\ref{def:multilog}, which equal zero at $k_1=0$ and $k_1=\pm 1/2$ and such that
\[ \eu^{-\iu \, g_{22}(k_1)/2} \eu^{-\iu \, g_{21}(k_1)/2} \gamma_2(k_1) \eu^{-\iu \, g_{21}(k_1)/2} \eu^{-\iu \, g_{22}(k_1)/2} = \Id. \]
Using the same function $\eta\sub{p}$ as before we can construct the corresponding 
\[ \widetilde{g}_{21}(k_1,k_2) = g_{11}(k_1) \, \eta\sub{p}(k_2-1/2) \quad \text{and} \quad \widetilde{g}_{22}(k_1,k_2) = g_{12}(k_1)\, \eta\sub{p}(k_2-1/2), \]
and define
\[ \alpha_2(\kk) := \eu^{-\iu \, \widetilde{g}_{22}(\kk)/2} \, \eu^{-\iu \, \widetilde{g}_{21}(\kk)/2}\, \alpha_1(\kk) \, \eu^{-\iu \, \widetilde{g}_{21}(\kk)/2} \, \eu^{-\iu \, \widetilde{g}_{22}(\kk)/2}. \]
This new family will again be a family of matching matrices, and
\[ \alpha_2(k_1,0) = \alpha_2(k_1,\pm 1/2) = \Id, \quad k_1\in \R. \]

Let
\[ K_1:=(0,1/2) \times (0,1/2), \quad K_2:=(0,1/2) \times (-1/2,0). \]
By two other successive constructions, we arrive at $\alpha_4(k_1,k_2)$ which equals the identity on the contour defined by $\partial K_1 \cup \partial (-K_1)\cup \partial K_2 \cup \partial (-K_2)$. Then we can write $\alpha_4(\kk) = m(\kk)\Id + \sum_{j=1}^{3} F_j(\kk) \, \sigma_j$ with $m\equiv1$ on the above contour. 

Let $I\subset (K_1\cup K_2)$ be a compact set such that $m(\kk)\geq 0$ on $\overline{K_1\cup K_2} \setminus I$. The vector field ${\bf F}(\kk)=[F_1(\kk),F_2(\kk),F_3(\kk)]$ is smooth, $\Z^2$-periodic, and odd. The restriction of ${\bf F}(\kk)$ to $I$ defines a smooth $2$-dimensional surface in $\R^3$. The origin in $\R^3$ is not an interior point of the range of ${\bf F}$, hence given $s>0$ one can find a vector ${\bf v}^{(s)}\in \R^3$ such that $\norm{{\bf v}^{(s)}}=s$ and $\norm{{\bf F}(\kk)+{\bf v}^{(s)}}>0$ on $I$. Because $I$ is compact, we can find $c_s>0$ such that
\[ \inf_{\kk\in I} \norm{{\bf F}(\kk)+{\bf v}^{(s)}} \geq c_s > 0. \]
Let $\chi\in C_0^\infty(\R^2)$ such that $0\leq \chi\leq 1$, $\chi=1$ on $I$ and ${\rm supp}(\chi)\subset (K_1\cup K_2)$. Define 
\[ {\bf F}_s(\kk) := {\bf F}(\kk) + (\chi(\kk)-\chi(-\kk)) {\bf v}^{(s)}, \quad \kk\in \bigcup_{j=1}^{2} \left (\overline{K_j}\cup \overline{-K_j}\right ), \]
and extend it to $\R^2$ by periodicity:
\[ {\bf F}_s(\kk)={\bf F}(\kk)+ \sum_{{\bf n}\in\Z^2}(\chi(\kk-{\bf n})-\chi(-\kk-{\bf n})){\bf v}^{(s)},\quad \kk\in \R^2. \]
Now we can apply the construction in \eqref{eqn:1Dalpha_eps} to $\alpha_4$ and obtain the desired multi-step logarithm for $\alpha$ in two additional steps.


\section{Construction of the multi-step logarithm: general rank} \label{sec:2d}

We come back to the general case of a $2$-dimensional family of matching matrices $\set{\alpha(\kk)}_{\kk \in \R^2}$. The aim of this Section is to show how to explicitly construct a multi-step logarithm for $\alpha$, assuming it is equivariantly null-homotopic (compare Theorem~\ref{TFAE}). We will actually prove also a stronger statement, which gives a continuous and periodic (but in general not time-reversal symmetric) multi-step logarithm for \emph{any} family of matching matrices, regardless of its equivariant homotopy class (compare Remarks~\ref{rmk:noTRS} and \ref{rmk:PeriodicFrames}).

In order to proceed with the construction of the multi-step logarithm, we will need to know what is the ``generic'' form of the spectrum of such families of matching matrices, much in the spirit of Theorem~\ref{thm:1Dhomotopies}. 

\begin{definition} \label{def:generic}
Let $\set{\alpha(\kk)}_{\kk \in \R^2}$ be an equivariantly null-homotopic family of matching matrices. We say that $\alpha$ is in \emph{generic form} if
\begin{enumerate}
 \item \label{item:genericB} $\eps = \Id$, and the spectrum of $\alpha(\kk)$ is completely non-degenerate for all $\kk \in \R^2$; or
 \item \label{item:genericF} $\eps = J$, and the spectrum of $\alpha(\kk)$ is completely non-degenerate in any compact set not containing the high-symmetry points $\kk_\sharp \in \R^2$ such that $\kk_\sharp \equiv -\kk_\sharp \bmod \Z^2$, it is \emph{doubly degenerate} at those points, and consists of clusters of two eigenvalues which are at distance at least $A$ from each other uniformly in the open balls of radius $R$ around each of the $\kk_\sharp$'s, where $A, R > 0$ are two positive constants (uniform in $\kk$).
\end{enumerate}
\end{definition}

The following result shows the origin of the terminology ``generic form''.

\begin{theorem} \label{thm:generic}
Assume that $d=2$. Let $\set{\alpha(\kk)}_{\kk \in \R^2}$ be an equivariantly null-homotopic family of matching matrices. Then one can construct a sequence $\set{\alpha_n(\kk)}_{\kk \in \R^2}$, $n \in \N$, of families of matching matrices such that
\begin{itemize}
 \item $\sup_{\kk \in \R^2} \norm{\alpha_n(\kk) - \alpha(\kk)} \to 0$ as $n \to \infty$, and
 \item for any $n \in \N$ the family $\set{\alpha_n(\kk)}_{\kk \in \R^2}$ is in generic form.
\end{itemize}
\end{theorem}

A stronger form of the above statement holds, if we drop the requirement of time-reversal symmetry for the approximants.

\begin{theorem} \label{thm:noTRS}
Assume that $d=2$. Let $\set{\alpha(\kk)}_{\kk \in \R^2}$ be a family of matching matrices. Then one can construct a sequence $\set{\widehat{\alpha}_n(\kk)}_{\kk \in \R^2}$, $n \in \N$, of continuous and $\Z^d$-periodic families of unitary matrices such that
\begin{itemize}
 \item $\sup_{\kk \in \R^2} \norm{\widehat{\alpha}_n(\kk) - \alpha(\kk)} \to 0$ as $n \to \infty$, and
 \item for any $n \in \N$ the matrix $\widehat{\alpha}_n(\kk)$ has completely non-degenerate spectrum for all $\kk \in \R^2$, which is moreover invariant under the exchange $\kk \mapsto -\kk$.
\end{itemize}
\end{theorem}

The proofs of the above Theorems are constructive but rather technical, and we defer them to the next Section. We conclude this Section by exhibiting the required multi-step logarithm. In the following we may assume that $\alpha(\kk) \in U(m)$ with $m \ge 3$, since in view of Remark~\ref{rmk:Properties} any family of matching matrices with $m=1$ admits a ``one-step'' logarithm in any dimension, and we have already treated the case $m=2$ in Section~\ref{sec:m=2}.

In the following, we denote by $\widehat{\alpha}$ any of the approximants provided by Theorem~\ref{thm:generic} or \ref{thm:noTRS} such that
\begin{equation} \label{eqn:<2}
\sup_{\kk \in \R^2} \norm{\widehat{\alpha}(\kk) - \alpha(\kk)} < 2.
\end{equation}

\begin{proposition} \label{prop:label}
Let $\Omega \subset \R^2$ be a star-shaped compact domain such that the family of unitary matrices $\widehat{\alpha}(\kk)$ has non-degenerate spectrum for $\kk \in \Omega$. Then the eigenvalues of $\widehat{\alpha}$ can be labeled so that they define smooth functions on $\Omega$.
\end{proposition}
\begin{proof}
Assume that every point $\kk$ of $\Omega$ can be connected to a fixed $\kk_0$ through a straight segment. Assume that $\widehat{\alpha}(\cdot)$ has $m$ non-degenerate eigenvalues. There exists a minimal distance $A>0$ between any two eigenvalues of $\widehat{\alpha}(\cdot)$ (compare Lemma~\ref{lemma2} below for $n=1$). 

We can label the eigenvalues of $\widehat{\alpha}(\kk_0)$ using the increasing order of their arguments $0\leq \phi_1(\kk_0) < \cdots < \phi_m(\kk_0) < 2\pi$.  Consider the family $\gamma(t) := \widehat{\alpha}((1-t) \, \kk_0 + t\, \kk)$, with $0\leq t\leq 1$. It is uniformly continuous on $[0,1]$ because $\widehat{\alpha}$ is uniformly continuous on $\Omega$. Hence there exists $\delta(A)>0$ such that 
\[ \norm{\gamma(t)-\gamma(t')} \leq A/100 \quad \text{whenever} \quad |t-t'| \leq \delta(A). \]
The spectrum of $\gamma(t)$ lies at a distance less than $A/10$ from the spectrum of $\gamma(t')$, for every $|t-t'| \leq \delta(A)$. Note that $\delta(A)$ does not depend on $\kk$.

Consider the Riesz projections of $\gamma(0)$ given by 
$$P_j(0)=\frac{1}{2\pi \iu} \int_{|z-\lambda_j(0)|=A/2}(z \Id -\gamma(0))^{-1} dz, \quad 1\leq j\leq N.$$ 
The spectrum of $\gamma(t)$ lies at a distance less than $A/10$ from the spectrum of $\gamma(0)$, for every $0 \leq t \leq \delta(A)$. Then the formulas
\[ P_j(t) = \frac{1}{2 \pi \iu} \int_{|z-\lambda_j(0)| = A/2} (z \Id- \gamma(t))^{-1} \, \di z, \quad 1 \leq j \leq m, \quad 0 \leq t \leq \delta(A) \]
give a smooth extension of the spectral projections of $\gamma(t)$, and $\lambda_j(t):=\Tr (P_j(t) \gamma(t))$ are its labeled eigenvalues. We can then repeat the construction on the interval $(\delta(A),2\delta(A)]$ starting from the labeling at $t=\delta(A)$. After a finite number of steps we reach $t=1$. 

Now define $\widetilde{\gamma}(t) = \widehat{\alpha}((1-t) \, \kk_0 + t \, \widetilde{\kk})$. If $\|\kk-\widetilde{\kk}\|$ is small enough, then 
\[ \sup_{0\leq t\leq 1} \norm{\gamma(t)-\widetilde{\gamma}(t)} \leq A/100. \]
This shows that around each eigenvalue of $\gamma(t)$ there exists exactly one eigenvalue of $\widetilde{\gamma}(t)$ which is situated sufficiently close to it. If $\lambda_j(1) \in \sigma(\gamma(1))$ and $\widetilde{\lambda}_{j'}(1) \in \sigma(\widetilde{\gamma}(1))$ are closer than $A/10$, then
\[
\left| \lambda_j(1-\delta) - \widetilde{\lambda}_{j'}(1-\delta) \right| \le \left| \lambda_j(1-\delta) - \lambda_{j}(1) \right| + \left| \lambda_{j}(1) - \widetilde{\lambda}_{j'}(1) \right| + \left| \widetilde{\lambda}_{j'}(1) - \widetilde{\lambda}_{j'}(1-\delta) \right| < \frac{3 A}{10}
\]
so that $\lambda_j(1-\delta)$ and $\widetilde{\lambda}_{j'}(1-\delta)$ are also close for the same choice of indices $j$ and $j'$. Since we know that there is only one eigenvalue of $\widetilde{\gamma}(1-\delta)$ which is close to the $j$-th eigenvalue of $\gamma(1-\delta)$, we deduce that it must be the $j'$-th, and we can then deduce that they are actually closer than $A/10$. We can then bootstrap the above argument all the way to $t=0$, where then $\lambda_j(0)$ and $\widetilde{\lambda}_{j'}(0)$ must coincide, thus implying that $j'=j$. Hence the labeling coincides for all $t \in [0,1]$ because it coincided at $t=0$.

In this way we also globally label the (Riesz) spectral projections, and they are as smooth as $\widehat{\alpha}$. Thus $\lambda_j(\kk)={\rm Tr}(P_j(\kk)\widehat{\alpha}(\kk))$, $j \in \set{1, \ldots, m}$, are also smooth.
\end{proof}

\begin{proposition} \label{prop:cut}
Let $\set{\widehat{\alpha}(\kk)}_{\kk \in \R^2}$ be an approximation of $\alpha$ as in Theorem~\ref{thm:generic} such that \eqref{eqn:<2} holds. Then one can construct a continuous and $\Z^2$-periodic function $\Lambda \colon \R^2 \to \R$ such that $\eu^{\iu \, \Lambda(\kk)}$ always lies in the resolvent set of $\alpha_n(\kk)$ for all $\kk \in \R^2$. If instead $\set{\widehat{\alpha}(\kk)}_{\kk \in \R^2}$ is as in Theorem~\ref{thm:noTRS}, the function $\Lambda$ can be chosen to be even as well.
\end{proposition}
\begin{proof}
By symmetry, it suffices to define $\Lambda$ continuously on the half unit cell $[0,1/2] \times [-1/2,1/2]$. Evenness and periodicity dictate how the function should be extended to the whole $\R^2$.

The bosonic case and the case in which there is no time-reversal symmetry are simpler to treat, since in these cases $\widehat{\alpha}$ has non-degenerate spectrum everywhere. Using the continuous labeling of the eigenvalues provided by Proposition~\ref{prop:label}, we can simply define
\[ \Lambda(\kk) := \frac{\phi_2(\kk) + \phi_3(\kk)}{2}, \]
where $\phi_j(\kk)$, $j \in \set{2,3}$ are continuous choices of the arguments for $\lambda_j(\kk)$.

The fermionic case is instead slightly more involved. The above definition works in the compact, star-shaped region
\[ \Omega := \left( \left[0,\frac12\right] \times \left[-\frac12,\frac12\right] \right) \setminus \bigcup_{i=1}^{4} B_R(\kk_i), \]
where the points $\kk_i$ are defined in \eqref{eqn:high-symmetry} and $R>0$ is as in Definition~\ref{def:generic}. In $\Omega$, the eigenvalues of $\widehat{\alpha}(\kk)$ are non-degenerate and can be labelled continuously. Instead, in the balls of radius $R$ around the points $\kk_i$, pairs of eigenvalues will cluster and become the doubly degenerate eigenvalues of $\widehat{\alpha}(\kk_i)$, but the clusters themselves stay at a positive distance $A>0$ from each other. Now, since $\widehat{\alpha}$ and $\alpha$ are close, also the GP-indices of the four restrictions of $\widehat{\alpha}$ which determine its homotopy classes will vanish \cite[Prop.~5.2]{CorneanMonacoTeufel17}, in view of the hypothesis in Theorem~\ref{thm:generic} that $\alpha$ is equivariantly null-homotopic (compare Theorem~\ref{thm:2Dhomotopies}). This implies in particular that, moving along these four directions, the eigenvalues which ``move out'' of one of the balls $B_R(\kk_i)$ have to come together again and form another cluster in any of the other balls $B_R(\kk_j)$. Thus, we may assume that the extensions to these balls of the eigenvalues labeled in $\Omega$ as $\set{\lambda_1, \lambda_2}$ form one cluster, while $\set{\lambda_3, \lambda_4}$ form a different cluster. As a consequence, the function $\Lambda$ defined above approaches the balls $B_R(\kk_i)$ while staying between two separate clusters: since there is a minimal positive distance $A>0$ between clusters, its definition can be extended continuously inside $B_R(\kk_i)$ while staying in the resolvent set of $\widehat{\alpha}(\kk)$, for example by staying at fixed distance from one of the clusters.
\end{proof}

With the tools above, we can finally prove

\begin{theorem} \label{thm:m-s-log}
Let $\set{\alpha(\kk)}_{\kk \in \R^2}$ be a family of matching matrices. Then one can construct a continuous and $\Z^2$-periodic multi-step logarithm for $\alpha$.

If moreover $\alpha$ is equivariantly null-homotopic, then the multi-step logarithm can be chosen to be also time-reversal symmetric.
\end{theorem}
\begin{proof}
As above, we denote by $\widehat{\alpha}$ any of the approximants provided by Theorem~\ref{thm:generic} or \ref{thm:noTRS} such that \eqref{eqn:<2} holds. From Proposition~\ref{prop:cut} we end up with a continuous, periodic, and possibly even function $\Lambda \colon \R^2 \to \R$ such that $\eu^{\iu \Lambda(\kk)}$ always lies in the resolvent set of $\widehat{\alpha}(\kk)$, for all $\kk \in \R^2$. This means that $-1$ is always in the resolvent set of the family of unitary matrices defined by $\widetilde{\alpha}(\kk) := \eu^{-\iu \Lambda(\kk)} \widehat{\alpha}(\kk)$, and consequently one can write $\widetilde{\alpha}(\kk) = \eu^{\iu \widetilde{h}(\kk)}$ for $\{\widetilde{h}(\kk)\}_{\kk \in \R^2}$ a continuous and $\Z^2$-periodic family of self-adjoint matrices, which moreover satisfies the time-reversal symmetry constraint \eqref{item:h_d} in Definition~\ref{def:multilog} if the original $\widehat{\alpha}$ is as in Theorem~\ref{thm:generic}. We conclude that
\[ \widehat{\alpha}(\kk) = \eu^{\iu h_2(\kk)}, \quad \text{with} \quad h_2(\kk) := \widetilde{h}(\kk) + \Lambda(\kk) \Id. \]
The family $\set{h_2(\kk)}_{\kk \in \R^2}$ still obeys the properties listed in Definition~\ref{def:multilog} (possibly with the exception of time-reversal symmetry), and moreover
\[ \sup_{\kk \in \R^2} \norm{\Id - \eu^{-\iu h_2(\kk)/2} \alpha(\kk) \eu^{-\iu h_2(\kk)/2}} < 2 \]
by \eqref{eqn:<2}. In turn this implies that $-1$ lies in the resolvent set of $\eu^{-\iu h_2(\kk)/2} \alpha(\kk) \eu^{-\iu h_2(\kk)/2}$ for all $\kk \in \R^2$, which gives that
\[ \eu^{-\iu h_2(\kk)/2} \alpha(\kk) \eu^{-\iu h_2(\kk)/2} = \eu^{\iu h_1(\kk)}, \quad \text{or equivalently} \quad \alpha(\kk) = \eu^{\iu h_2(\kk)/2} \eu^{\iu h_1(\kk)} \eu^{\iu h_2(\kk)/2} \]
for a family of matrices $\set{h_1(\kk)}_{\kk \in \R^2}$ as in Definition~\ref{def:multilog} (again with the possible exception of time-reversal symmetry). We recognize that the above is exactly \eqref{multilog} for $M=2$, and hence the desired multi-step logarithm for $\alpha$ has been constructed.
\end{proof}


\section{Technicalities about the generic form}

This Section is devoted to the proofs of Theorems~\ref{thm:generic} and \ref{thm:noTRS}. We let $\set{\alpha(\kk)}_{\kk \in \R^2}$ be a $2$-dimensional family of matching matrices, and we want to construct approximants to $\alpha$ which are in generic form. Due to the $\Z^2$-periodicity and time-reversal symmetry of families of matching matrices, which reflect in periodicity and evenness of their spectra, it suffices to perform the construction on the half unit cell $C = [0,1/2] \times [-1/2,1/2]$ in such a way that a periodic, time-reversal symmetric extension results in a smooth family.

\subsection{Local splitting lemma}

The first step to put $\alpha$ in generic form is to control the spectrum at the high-symmetry points $\kk_\sharp$ such that $\kk_\sharp \equiv - \kk_\sharp \bmod \Z^2$. This is achieved through the following Lemma, which is valid in any dimension and generalizes \cite[Lemma~A.1]{CorneanMonacoTeufel17}.

\begin{lemma}[Local Splitting Lemma] \label{lemma:splitting}
Let $R>0$ and $\kk_\sharp \in \R^d$ be a high-symmetry point. Denote by $B_R(\kk_\sharp)$ the open ball of radius $R$ centered at $\kk_\sharp$. Let $\set{\alpha(\kk)}_{\kk \in B_R(\kk_\sharp)}$ be a continuous and time-reversal symmetric family of unitary matrices. Then it is possible to construct a sequence $\set{\alpha_n(\kk)=\eu^{\iu h_n(\kk)}}_{\kk \in B_{R'}(\kk_\sharp)}$, possibly for $0 < R' \le R$,  of continuous and time-reversal symmetric families of unitary matrices, with $\kk \mapsto h_n(\kk)$ as in Definition~\ref{def:multilog}, such that
\[ \lim_{n \to \infty} \sup_{\kk \in B_{R'}(\kk_\sharp)} \norm{\alpha_n(\kk) - \alpha(\kk)} = 0, \]
and
\begin{itemize}
 \item the spectrum of $\alpha_n(\kk_\sharp)$ is completely non-degenerate if $\eps = \Id$, or 
 \item each eigenvalue of $\alpha_n(\kk_\sharp)$ is \emph{doubly degenerate} if $\eps = J$.
\end{itemize}

Moreover, it is possible to construct a sequence $\set{\alpha_n(\kk)=\eu^{\iu \widehat{h}_n(\kk)}}_{\kk \in B_{R'}(\kk_\sharp)}$, possibly for $0 < R' \le R$,  of continuous families of unitary matrices, with $\kk \mapsto \widehat{h}_n(\kk)$ satisfying \eqref{item:h_a}, \eqref{item:h_b} and \eqref{multilog} in Definition~\ref{def:multilog}, and such that
\begin{itemize}
 \item $\sup_{\kk \in B_{R'}(\kk_\sharp)} \norm{\widehat{\alpha}_n(\kk) - \alpha(\kk)} \to 0$ as $n \to \infty$, and
 \item the spectrum of $\widehat{\alpha}_n(\kk)$ is completely non-degenerate and even-symmetric with respect to $\kk_\sharp$ for $\kk \in B_{R'}(\kk_\sharp)$.
\end{itemize}
\end{lemma}
\begin{proof}
By a shift $\kk \to \kk-\kk_\sharp$, we can always assume $\kk_\sharp = \mathbf{0}$. Notice indeed that the combination of periodicity and time-reversal symmetry implies that the spectrum of a family of matching matrices is even-symmetric around any of the high-symmetry points (that is, under the transformation $\kk_\sharp + \mathbf{q} \mapsto \kk_\sharp - \mathbf{q}$).

\smallskip

\noindent \textsl{Step 1: The local logarithm.} Assume that the spectrum of $\alpha(\mathbf{0})$ consists of $1\leq p_0\leq  n$ (possibly even degenerate, due to Kramers degeneracy) eigenvalues labeled as $\{\lambda_1(\mathbf{0}),\ldots,\lambda_{p_0}(\mathbf{0})\}$ in the increasing order of their principal arguments. If $s>0$ is sufficiently small and $\norm{\kk}<s$, due to the continuity of $\kk \mapsto \alpha(\kk)$ we know that the spectrum of $\alpha(\kk)$ will also consist of well separated clusters of eigenvalues. Let $\Pi_j(\kk)$ be the spectral projection of $\alpha(\kk)$ corresponding to the $j$-th cluster:
\[ \Pi_j(\kk)=\frac{1}{2\pi \iu}\int_{C_j}(z \Id -\alpha(\kk))^{-1}\,\di z,\quad \eps \, \Pi_j(\kk)= \Pi_j(-\kk)^t\,\eps,\quad \norm{\kk}<s. \]
The matrix $\alpha(\kk)$ is block diagonal with respect to the decomposition $\mathbb{C}^{m}=\bigoplus_{j=1}^{p_0} \Pi_j(\kk)\mathbb{C}^{m}$, {\sl i.e.}  $\alpha(\kk)=\sum_{j=1}^{p_0} \Pi_j(\kk)\alpha(\kk)\Pi_j(\kk)$, $\norm{\kk}<s$. 

Define $$\widetilde{\alpha}(\kk):=\sum_{j=1}^{p_0} \lambda_j(\mathbf{0})\Pi_j(\kk)=\exp \left(\iu \, \sum_{j=1}^{p_0} {\rm Arg}(\lambda_j(\mathbf{0}))\Pi_j(\kk) \right),\quad \norm{\kk}<s.$$
The matrix $\widetilde{\alpha}(\kk)$ is unitary, commutes with $\alpha(\kk)$, and $\eps \, \widetilde{\alpha}(\kk) = \widetilde{\alpha}(-\kk)^t \,\eps$ if $\norm{\kk}<s$. Define $\gamma(\kk):=\widetilde{\alpha}^{-1}(\kk)\alpha(\kk)$; we have that $\gamma(\kk)$ is unitary, commutes with $\alpha(\kk)$,  $\eps \, \gamma(\kk) = \gamma(-\kk)^t \, \eps$ if $\norm{\kk}<s$ and $\lim_{\kk\to \mathbf{0}}\gamma(\kk)=\Id$. In particular, $-1$ is never in the spectrum of $\gamma(\kk)$. Going through the Cayley transform we can find a self-adjoint matrix $\widetilde{h}(\kk)$ such that 
\[ \gamma(\kk)=\eu^{\iu \widetilde{h}(\kk)},\quad \widetilde{h}(\mathbf{0})=0,\quad \eps \, \widetilde{h}(\kk) = \widetilde{h}(-\kk)^t \, \eps,\quad [\Pi_j(\kk),\widetilde{h}(\kk)]=0,\quad \norm{\kk}<s. \]
We obtain
\[ \alpha(\kk)=\exp \left( \iu\left(\widetilde{h}(\kk) + \sum_{j=1}^{p_0} {\rm Arg}(\lambda_j(\mathbf{0}))\Pi_j(\kk)\right) \right),\quad \norm{\kk}<s. \]

In the following, we will then assume that $\alpha(\kk)=\eu^{\iu h(\kk)}$ on $B_R(\mathbf{0})$, at the expense of taking a smaller $R$.

\smallskip

\noindent \textsl{Step 2: Splitting of spurious degeneracies.} Let us use the generic notation $\Pi$ for any of the spectral projections onto one of the eigenvalues of $\alpha(\mathbf{0})$; let $p$ denote the dimension of the range of $\Pi$ ({\sl i.e.} the degeneracy of the eigenvalue). 

We distinguish now the bosonic and the fermionic case. When time-reversal symmetry is of bosonic type, the choice of any \emph{real} orthonormal basis for $\Pi = \Pi^t$ gives a decomposition
\begin{equation}\label{noTRS}
\Pi=\bigoplus_{j=1}^{p} P_{j},\quad  \dim \Ran P_{j}=1,\quad P_{j}=P_{j}^2=P_j^*=P_{j}^t.
\end{equation}

When we have instead fermionic time-reversal symmetry, we proceed as follows. Denote by $C$ the complex conjugation with respect to the basis in which $\eps$ is represented by the symplectic matrix $J$, and by $\Theta:=\eps C$; then $\Theta$ is an antiunitary operator satisfying $\Theta^2=-\Id$ on $\Ran \Pi \simeq \C^{m}$. From \eqref{alphaTRS} we see that we have the property
\[ \Theta \Pi =\Pi \Theta,\quad \text{or}\quad \eps \Pi= \Pi^t\eps,\quad \dim\Ran \Pi =2r. \]
We want to prove the following decomposition formula:
\begin{equation}\label{now4}
\Pi=\bigoplus_{j=1}^{r} P_{j},\quad  \dim \Ran P_{j}=2,\quad P_{j}=P_{j}^*=P_{j}^2,\quad \eps P_{j}=P_{j}^t\eps.
\end{equation}
Start by choosing an arbitrary unit vector $\mathbf{v}_1\in \Ran \Pi$. Define $\mathbf{v}_2=\Theta \mathbf{v}_1$. We have that $\Pi \mathbf{v}_2=\Pi\Theta \mathbf{v}_1=\Theta \Pi \mathbf{v}_1=\mathbf{v}_2$, hence $\mathbf{v}_2$ also belongs to $\Ran \Pi$. We also know that $\scal{\mathbf{v}_1}{\mathbf{v}_2}=0 $ and $\mathbf{v}_1=-\Theta \mathbf{v}_2$. Define 
 $$P_1=\ket{\mathbf{v}_1} \bra{\mathbf{v}_1} + \ket{\mathbf{v}_2} \bra{\mathbf{v}_2}.$$
Let $\mathbf{f}\in \C^{m}$. Then
\begin{align*}
\Theta P_1 \mathbf{f} &=\mathbf{v}_2 \overline{\scal{\mathbf{v}_1}{\mathbf{f}} }-\mathbf{v}_1\overline{\scal{ \mathbf{v}_2}{\mathbf{f}} }=
-\mathbf{v}_1\scal{\mathbf{f}}{\mathbf{v}_2} +\mathbf{v}_2 \scal{\mathbf{f}}{\mathbf{v}_1}=-\mathbf{v}_1\scal{ \Theta \mathbf{v}_2}{\Theta \mathbf{f}} +\mathbf{v}_2\scal{ \Theta \mathbf{v}_1}{\Theta \mathbf{f}}\\
&=\mathbf{v}_1\scal{ \mathbf{v}_1}{\Theta \mathbf{f}}+\mathbf{v}_2\scal{ \mathbf{v}_2}{\Theta \mathbf{f}}=P_1\Theta \mathbf{f}.
\end{align*}
This shows that $\Theta P_1=P_1\Theta$, or $\eps P_{1}=P_{1}^t\eps$. If $r>1$ we continue inductively. Let $\mathbf{v}_3$ be an arbitrary unit vector orthogonal to $\Ran P_1$ in $\Ran \Pi$. Define $\mathbf{v}_4=\Theta \mathbf{v}_3$. As before, we can show that $\Pi \mathbf{v}_4=\mathbf{v}_4$ and $\scal{ \mathbf{v}_3}{\mathbf{v}_4} =0$.  Moreover, using the properties of $\Theta$ we can show that $\mathbf{v}_4$ is also orthogonal on both $\mathbf{v}_1$ and $\mathbf{v}_2$. Define $P_2=\ket{\mathbf{v}_3}\bra{\mathbf{v}_3} + \ket{\mathbf{v}_4}\bra{\mathbf{v}_4}$. The proof of $\Theta P_2=P_2\Theta$ is the same as the one for $P_1$. We now continue inductively until we exhaust the range of $\Pi$. We conclude that \eqref{now4} is proved. 

Now let us apply now \eqref{noTRS} (respectively \eqref{now4}) to each spectral projection $\Pi_j(\mathbf{0})$ of $\alpha(\mathbf{0})$. Assume that the multiplicity of the eigenvalue $\lambda_j(\mathbf{0})$ ({\sl i.e.} the dimension of $\Ran \Pi_j(\mathbf{0})$) equals $p_j(\mathbf{0})$, where $1\leq p_j(\mathbf{0})\leq m$. In presence of fermionic time-reversal symmetry, we have that this multiplicity is even, $p_j(\mathbf{0}) = 2 r_j(\mathbf{0})$.
 
Then \eqref{noTRS} leads to
\begin{equation} \label{zumba9noTRS}
\Pi_j(\mathbf{0})=\bigoplus_{l_j=1}^{p_j(\mathbf{0})} P_{j,l_j}(\mathbf{0}),\quad \dim\Ran P_{j,l_j}(\mathbf{0})=1,\quad P_{j,l_j}(\mathbf{0})=P_{j,l_j}(\mathbf{0})^*=P_{j,l_j}(\mathbf{0})^2=P_{j,l_j}(\mathbf{0})^t.
\end{equation}
Respectively \eqref{now4} leads to
\begin{align}\label{zumba9}
\Pi_j(\mathbf{0})=\bigoplus_{l_j=1}^{r_j(\mathbf{0})} P_{j,l_j}(\mathbf{0}),\quad \dim\Ran P_{j,l_j}(\mathbf{0})=2,\quad P_{j,l_j}(\mathbf{0})=P_{j,l_j}(\mathbf{0})^*=P_{j,l_j}(\mathbf{0})^2,\quad \eps P_{j,l_j}(\mathbf{0})=P_{j,l_j}(\mathbf{0})^t\eps.
\end{align}
Define
\[ A_j(\mathbf{0}):=\sum_{l_j=1}^{m_j(\mathbf{0})} (l_j-1)\; P_{j,l_j}(\mathbf{0}) \]
where $m_j(\mathbf{0}) := p_j(\mathbf{0})$ in the bosonic case, and $m_j(\mathbf{0}) := r_j(\mathbf{0})$ in the fermionic case. From \eqref{zumba9noTRS} and \eqref{zumba9}, we also have $\eps A_j(\mathbf{0})=A_j(\mathbf{0})^t\eps$, where $\eps = \Id$ or $\eps = J$ depending on the type of time-reversal symmetry. Seen as an operator acting on $\Ran\Pi_j(\mathbf{0})$, the spectrum of $A_j(\mathbf{0})$ consists of (doubly degenerate in the fermionic case) eigenvalues given by $\{0,1,...,m_j(\mathbf{0})-1\}$. Of course, if $r_j(\mathbf{0})=1$ then $A_j(\mathbf{0})=0$. 

Now let $g_s\colon\R^d\to [0,1]$ be smooth, even, $g_s(\mathbf{0})=1$, and $\mathrm{supp}(g_s)\subset B_{s/2}(\mathbf{0})$. Define
\[ v_s(\kk):= s \, g_s(\kk) \left (\sum_{j=1}^{p_0} A_j(\mathbf{0}) \right) \]
where $p_0$ is the total number of distinct eigenvalues of $\alpha(\mathbf{0})$. Then the support of $v_s$ is contained in $B_s(\mathbf{0})$ and obeys time-reversal symmetry. Define $\alpha_s(\kk)$ to be $ \eu^{\iu(h(\kk)+v_s(\kk))}$ if $\kk\in B_s(\mathbf{0})$, and let  $\alpha_s(\kk)=\alpha(\kk)$ outside $B_s(\mathbf{0})$. The family $\alpha_s$ obeys all the required properties, at the price of taking possibly a smaller $s$ to avoid overlapping of the eigenvalues at $\mathbf{0}$, and converges uniformly in norm to $\alpha$ when $s\to 0$. 

\smallskip

\noindent \textsl{Step 3: Complete splitting.} We come to the final part of the statement. By virtue of the above constructions, this is of relevance only if the original family of matching matrices $\alpha$ is of fermionic nature, and we will assume that $\alpha(\kk)$ is defined on a ball $B = B_s(\mathbf{0})$ of radius $s>0$ around ${\bf 0}$, has doubly degenerate eigenvalues at $\mathbf{0}$, and has non-degenerate spectrum on the surface $S_s$ of this ball. In view of time-reversal symmetry, the spectrum will be also even-symmetric. For $\kk \in S_s$, write the spectral decomposition of $\alpha(\kk)$ as
\[ \alpha(\kk) = \sum_{j=1}^{m} \lambda_j(\kk) \, P_j(\kk), \quad \dim \Ran P_j(\kk) = 1. \]
Notice that $\lambda_j(\kk) = \lambda_j(-\kk)$ by time-reversal symmetry. 

Say that the eigenvalues $\lambda_1, \lambda_2$ belong to the cluster which originates from some doubly degenerate eigenvalue $\lambda = \eu^{\iu \phi}$ of $\alpha(\mathbf{0})$, and assume that their are labelled by the increasing order of the values of their principal arguments%
\footnote{Notice that the choice of principal arguments together with time-reversal symmetry imply $\phi_j(\kk) = \phi_j(-\kk)$, $j \in \set{1,2}$, $\kk \in S_s$.} %
$\phi_1 < \phi_2$. Consider the spectral projection $\Pi(\mathbf{0})$ of $\alpha(\mathbf{0})$ relative to the eigenvalue $\lambda$, and call $\Pi(\kk) = P_1(\kk) \oplus P_2(\kk)$ for $\kk \in S_s$. If $s$ is sufficiently small, then $\norm{\Pi(\kk) - \Pi(\mathbf{0})} < 1$, and the two projections are related by a Kato-Nagy unitary $U_s(\kk)$. Moreover, since $U_s(\kk)$ is close to the identity, we can write $U_s(\kk) = \eu^{\iu H_s(\kk)}$ for a smooth family of self-adjoint matrices $H_s(\kk)$. The same is true when we replace $\kk$ with $-\kk$. Let 
\[ V_s(t \kk) := \begin{cases} 
\eu^{\iu t H_s(\kk)} & \text{if } t \in [0,1] \\
\eu^{\iu |t| H_s(-\kk)} & \text{if } t \in [-1,0] \\
\end{cases} \] 
be a smooth path of unitaries joining $U_s(\kk)$ to $U_s(-\kk)$, passing through the identity at $\kk = \mathbf{0}$.

A choice of a basis in $\Ran \Pi(\mathbf{0})$ corresponds to a splitting $\Pi(\mathbf{0}) = P_1(\mathbf{0}) \oplus P_2(\mathbf{0})$, with $\dim \Ran P_j(\mathbf{0}) = 1$. Setting
\[ P_j^{(s)}(t\kk) := V_s(t\kk) P_j(\mathbf{0}) V_s(t\kk)^{-1}, \quad j \in {1,2}, \quad t \in [-1,1]\]
gives a path of $1$-dimensional projections joining $P_j(-\kk)$ with $P_j(\kk)$, and coinciding with the chosen $P_j(\mathbf{0})$ at $t=0$. 

Finally, if $g>0$ denotes the minimal distance between the degenerate eigenvalues of $\alpha(\mathbf{0})$, denote by $\lambda_j^{(s)}(t\kk) = \eu^{\iu \, \phi_j^{(s)}(t \kk)}$, where
\[ \phi_j^{(s)}(t \kk) := (1-|t|) \left(\phi + (-1)^j \dfrac{g}{10} \right) + |t| \, \phi_j(\kk), \quad j \in \set{1,2}, \quad t \in [-1,1]. \]
For example, the function $\lambda_2^{(s)}(t \kk)$ interpolates between $\eu^{\iu (\phi + g/10)}$ and $\lambda_2(\kk)$ for positive $t \in [0,1]$, and between $\lambda + g/10$ and $\lambda_2(-\kk)$ for negative $t \in [-1,0]$; the function $\lambda_1^{(s)}(t \kk)$ does the same, but starting from $\eu^{\iu (\phi - g/10)}$. Since at the level of principal arguments the interpolations are linear, we have that $\phi_1^{(s)}(t \kk) < \phi_2^{(s)}(t \kk)$ for all $t \in [-1,1]$.

We repeat the above construction for all the $n = m/2$ pairs of spectral subspaces originating from the doubly degenerate eigenvalues of $\alpha(\mathbf{0})$. We set then for $\kk \in S_s$
\[ \alpha_s(t \kk) : = \sum_{j=1}^{m} \lambda_j^{(s)}(\kk) P_j^{(s)}(t\kk), \quad t \in [-1,1] \]
on the ray joining $\kk$ to $-\kk$ through $\mathbf{0}$. The above gives the required approximation of $\alpha$ with non-degenerate, even-symmetric spectrum.
\end{proof}

\subsection{Extending non-degeneracy to 1D} 

The next step in the construction of the approximants requires to remove multiple eigenvalues which are not due to Kramers degeneracy. In this Subsection we will do so on a ``quasi $1$-dimensional'' set of points in $\R^2$, for which one of the two coordinates is very close to a half-integer.

In view of the Local Splitting Lemma \ref{lemma:splitting}, we can assume that the matrices $\alpha(\kk_\sharp)$, for $\kk_\sharp$ a high-symmetry point, have non-degenerate spectrum (in the bosonic case) or only doubly degenerate eigenvalues (in the fermionic case). Moreover, since by the assumption in Theorem~\ref{thm:generic} the family $\alpha$ is equivariantly null-homotopic, we know that its restrictions to the four lines $\set{k_{1,2} = 0}$ are $\set{k_{1,2} = 1/2}$ are equivariantly null-homotopic (compare Theorem~\ref{thm:2Dhomotopies}); by periodicity the same is true on any line where one of the coordinates is equal to a half-integer. In particular, in the fermionic case this implies that the GP-indices of the restrictions to these lines vanish (Corollary~\ref{crl:2Dhomotopies}).
 
\begin{proposition}
Under the above assumptions, there exists a sequence of families of matching matrices $\set{\alpha_\ell(\kk)}_{\kk \in \R^2}$ such that
\[ \lim_{\ell \to \infty} \sup_{\kk \in \R^2} \norm{\alpha_\ell(\kk) - \alpha(\kk)} = 0, \]
and such that for all $\ell \in \N$
\begin{itemize}
 \item the restriction of the family $\set{\alpha_\ell(\kk)}_{\kk \in \R^2}$ to the lines $\set{k_j = p_j/2}$, $j \in \set{1,2}$, $p_j \in \Z$, is completely non-degenerate if $\eps = \Id$, or
 \item the same restrictions have completely non-degenerate spectrum on any given compact interval not containing the double Kramers degeneracies if $\eps = J$.
\end{itemize}
\end{proposition}
\begin{proof}
Fix a compact set $I\subset (-1/2,0)\cup (0,1/2)$ symmetric with respect to $0$. Since $\set{\alpha(k_1,0)}_{k_1 \in \R}$ is equivariantly null-homotopic, it can be approximated arbitrarily well with a smooth $1$-dimensional family of matching matrices $\set{\widetilde{\alpha}_1(k_1)}_{k_1 \in \R}$ which has completely non-degenerate spectrum in $I$ and which admits a traceless, smooth and periodic logarithm $h_1(k_1)$ such that $\widetilde{\alpha}_1(k_1)=\eu^{\iu h_1(k_1)}$ and $\eps \, h_1(k_1) = h_1(-k_1)^t \, \eps^{-1}$: this is Theorem~\ref{thm:1Dhomotopies}\eqref{item:EquivariantB} (compare \cite[Lemma~2.18]{CorneanHerbstNenciu16}) in the bosonic case, where $I$ can be actually taken to be the whole interval $[-1/2,1/2]$, while it is Theorem~\ref{thm:1Dhomotopies}\eqref{item:EquivariantF} (compare \cite[Prop.~5.4(2)]{CorneanMonacoTeufel17}) in the fermionic case, under the assumption mentioned above that $\rueda(\alpha(\cdot,0)) = 0 \in \Z_2$. 

Define the unitary matrix
\[ \gamma_1(k_1,k_2):=\eu^{-\iu  h_1(k_1)/2}\alpha(k_1,k_2)\eu^{-\iu  h_1(k_1)/2}. \]
If $k_2$ lies closer than some $\delta_0>0$ from any given integer, then  $\gamma_1(k_1,k_2)$ is close to the identity matrix. Via the Cayley transform we can construct a traceless selfadjoint $H_1(k_1,k_2)$ with $\eps \, H_1(k)=H_1(-k)^t \, \eps$ when $k_2$ is closer than $\delta_0$ to an integer. The family $H_1$ is periodic in $k_1$ and also in $k_2$ near the integers. 

Now let  $g_\delta \colon \R \to \R$, $0\leq g_\delta\leq 1$, be smooth, even, equal to $1$ on $[-\delta,\delta]$ and supported in the interval 
 $(-2\delta,2\delta)$ with $\delta$ small. Define $G_\delta(x)=\sum_{n\in \Z}g_\delta(x-n)$. If $\delta\leq \delta_0/10$ we define 
\[ \alpha'(k_1,k_2):=\eu^{\iu  h_1(k_1)/2}\eu^{\iu (1-G_\delta(k_2))H_1(k_1,k_2)}\eu^{\iu  h_1(k_1)/2} \]
if $k_2$ is closer than $3\delta$ from $\Z$, and $\alpha'(k_1,k_2)=\alpha(k_1,k_2)$ otherwise. 
 
The matrix $\alpha'(k_1,0)$ coincides with $\widetilde{\alpha}_1(k_1)$, hence it is non-degenerate on $k_1\in I$. If $\delta$ is small, $\alpha'(k_1,k_2)$ will continue to be equivariantly null-homotopic and we continue by investigating $\alpha'$. When we modify $\alpha'$ near $k_1=0$ by a similar construction as above,  we see that the affected region in $k_1$ is of order $\delta$, so by choosing $\delta$ small enough we do not change anything  in the segment $\set{k_1\in I, \:k_2=0}$, so the previous non-degeneracy is not affected. We thus obtain $\alpha''$ which has complete non-degeneracy on two segments: $\set{ k_1\in I, \: k_2=0}$ and $\set{k_1=0, \: k_2\in I}$. After four steps and putting $\delta=1/\ell$ (starting with some large enough $\ell_0$) we finish the construction of $\alpha_\ell$. 
\end{proof}

\begin{remark}
The proof above can be adapted, with minor modifications, to show the existence of approximants $\set{\widehat{\alpha}_\ell(\kk)}_{\kk \in \R^2}$ for \emph{any} (not necessarily equivariantly null-homotopic) family of matching matrices $\alpha$ which are continuous, $\Z^d$-periodic, and whose spectrum is completely non-degenerate on the lines $\set{k_j = p_j/2}$, $j \in \set{1,2}$, $p_j \in \Z$. One just needs to use the second part of the statement of Lemma~\ref{lemma:splitting} (together with \cite[Prop.~5.4(1)]{CorneanMonacoTeufel17}) in the fermionic case, to lift the Kramers degeneracies at the cost of breaking time-reversal symmetry. This goes into the proof of Theorem~\ref{thm:noTRS}.
\end{remark}

From now on we may assume that $\alpha$ is completely non-degenerate on some thin slabs centered around the lines $(k_1,0)$, $(0,k_2)$, $(k_1,1/2)$ and $(1/2,k_2)$, possibly not including some small open balls containing the crossing points of these lines in the fermionic case. Inside these balls, only double degeneracies are allowed. Due to the periodicity of $\alpha$, this property is extended around all the lines of the type  $\set{k_j = p_j/2}$,  with $j \in \set{1,2}$ and $p_j\in\Z$. Notice that it is only on these slabs that time-reversal symmetry plays the role of a compatibility condition.

In the following, we proceed with our construction of the approximants of $\alpha$ in generic form on the quadrant $Q = [0,1/2] \times [0,1/2]$; a similar construction works for $Q' = [0,1/2] \times [-1/2,0]$. Since the half unit cell $C = Q \cup Q'$, this will suffice in view of the considerations at the beginning of this Section.
  
Let $[a,b] \subset (0,1/2)$ be the interval obtained by projection on one of the axes of the complement of the thin slabs where $\alpha$ has non-degenerate spectrum. Let us consider the four segments given by 
\[ S_1:=[a,b]\times \{0\},\quad S_2:=[a,b]\times \{1/2\},\quad S_3:= \{0\}\times [a,b],\quad S_4:=\{1/2\}\times [a,b]. \]
The previous Proposition allows us to assume that $\alpha(k)$ is completely non-degenerate on a thin neighbourhood of these segments. We now we want to close the contour by keeping complete non-degeneracy. 

\begin{proposition}\label{propolinie}
Let $\Lambda$ be the segment joining the endpoint $(0,b)$ of $S_3$ with the endpoint $(a,1/2)$ of $S_2$. Then there exists a sequence of families of matching matrices $\set{\alpha_n(\kk)}_{\kk \in Q}$ such that
\begin{itemize}
 \item $\sup_{\kk \in Q} \norm{\alpha_n(\kk) - \alpha(\kk)} \to 0$ as $n \to \infty$, and
 \item for all $n \in \N$ the spectrum of $\alpha_n(k)$ is completely non-degenerate on a thin ($n$-dependent) compact neighbourhood of $\Lambda$.
\end{itemize}
\end{proposition}
 \begin{proof}
We introduce some natural global coordinates $(x_\lambda,y_\lambda)$ induced by $\Lambda \subset \R^2$ such that a point belonging to $\Lambda$ is represented by $y_\lambda=0$ and $x_\lambda$ belongs to an interval. 
  
When we restrict $\alpha(\kk)$ to the closed segment $\Lambda$ we obtain a $1$-dimensional family of unitary matrices. We know that we can approximate it arbitrarily well with a smooth family of unitary matrices with no degeneracies on $\Lambda$, which also admits a smooth  logarithm $h_n(x_\lambda)$, {\sl i.e.} 
\[ \sup_{\kk\in \Lambda} \norm{\eu^{-\iu h_n(x_\lambda)/2} \, \alpha(\kk(x_\lambda, 0)) \, \eu^{-\iu h_n(x_\lambda)/2}-\Id} \leq \frac{1}{n}. \]
We know that $\alpha$ is non-degenerate at the endpoints, hence by continuity it will remain non-degenerate on a small neighbourhood of them. Thus we may find a smaller open segment  $\Lambda'\subset \Lambda$ such that $\alpha(\kk)$ is non-degenerate on $\Lambda\setminus \Lambda'$. Then $\eu^{\iu h_n(x_\lambda)}$ will also be completely non-degenerate on $\Lambda\setminus \Lambda'$ if $n$ is large enough, while the minimal distance between its eigenvalues is bounded from below by some $n$-independent positive constant $A>0$. Of course, this might no longer be true inside $\Lambda'$. 
 
 Denote by $\Lambda'_r$ the convex open set containing all the points $\kk$ such that ${\rm dist}(\kk,\Lambda')<r$.  There exists some $r_n$ small enough such that
\begin{equation}\label{ianua1}
\sup_{\kk \in \overline{\Lambda'_{r_n}}} \norm{ \eu^{-\iu h_n(x_\lambda)/2}\, \alpha(\kk(x_\lambda, y_\lambda)) \, \eu^{-\iu h_n(x_\lambda)/2}-\Id} \leq \frac{2}{n}.
\end{equation}
We can also assume that $\overline{\Lambda'_{r_n}}$ is always included in the open square $(0,1/2)\times (0,1/2)$.

Consider the unitary operator
$$\gamma_n(\kk(x_\lambda,y_\lambda)):=\eu^{-\iu h(x_\lambda)/2}\, \alpha(\kk(x_\lambda, y_\lambda)) \,\eu^{-\iu h(x_\lambda)/2},\quad \quad \kk(x_\lambda, y_\lambda)\in\Lambda'_{r_n} . $$
From \eqref{ianua1} we conclude that $ \gamma_n(\kk(x_\lambda,y_\lambda))$ is close to the identity operator and it admits a smooth logarithm $H_n(\kk(x_\lambda,y_\lambda))$, hence 
$$\gamma_n(\kk(x_\lambda,y_\lambda))=\eu^{\iu  H_n(\kk(x_\lambda,y_\lambda))},\quad \kk(x_\lambda,y_\lambda)\in \Lambda'_{r_n},\quad \norm{H_n}=\mathcal{O}(1/n).$$

Consider a smooth function $0\leq \chi_n\leq 1$ which equals $1$ on  $\overline{\Lambda'_{r_n/10}}$ and has support on  $\Lambda'_{r_n/5}$. Define the following unitary matrix in the closed rectangle $Q$:
\[
\beta_n(\kk(x_\lambda,y_\lambda)):= \begin{cases}
  \eu^{\iu h_n(x_\lambda)/2}
 \eu^{\iu [1- \chi_n(\kk(x_\lambda,y_\lambda))] H_n(\kk(x_\lambda,y_\lambda))}\eu^{\iu h_n(x_\lambda)/2} & \text{if } \kk(x_\lambda,y_\lambda) \in  \Lambda'_{r_n} \\
   \alpha(k(x_\lambda,y_\lambda)) & \text{if } \kk(x_\lambda,y_\lambda)\in Q \setminus \Lambda'_{r_n}. \end{cases}
\]
We see that $\beta_n$ differs from $\alpha$ only on $\Lambda'_{r_n/5}$, and $\beta_n-\alpha =\mathcal{O}(1/n)$ anywhere in $Q$. Also, $\beta_n$ is smooth. 

Now let us show that $\beta_n$ is completely non-degenerate on $\Lambda$. We write
\[
\beta_n(\kk(x_\lambda,0)):= \begin{cases}
  \eu^{\iu h_n(x_\lambda)/2}
 \eu^{\iu [1-\chi_n(k(x_\lambda,0))] H_n(k(x_\lambda,0))}\eu^{\iu h_n(x_\lambda)/2} & \text{if }  \kk(x_\lambda,0) \in  \Lambda'_{r_n} \\
   \alpha(k(x_\lambda,0)) &  \text{if }  \kk(x_\lambda,0)\in  \Lambda\setminus \Lambda'_{r_n}.  \end{cases}
\]
We see that the only problems might come from the region corresponding to 
$$\Lambda \cap \{\Lambda'_{r_n/5}\setminus \overline{\Lambda'_{r_n/10}}\}\subset \Lambda\setminus \Lambda'.$$ 
In that region, $H_n$ is of order $1/n$ while the logarithm $h_n$ has a positive minimal splitting which is independent of $n$. Thus if $n$ is large enough, no degeneracies can be induced on $\Lambda\setminus \Lambda'$. Also, by continuity, the non-degeneracy on $\Lambda$ can be extended to an $n$ dependent slab containing $\Lambda$. This concludes the proof.
\end{proof}
 
We can now join the extremities of the fours segments $S_1, \ldots, \: S_4$ in order to obtain a closed polygonal line inside the rectangle $Q$. A similar construction as in Proposition \ref{propolinie} provides us with an $\alpha$ which is completely non-degenerate around a thin neighbourhood of this polygonal line. 

The next step is to avoid any eigenvalue crossings inside this polygon. 

\subsection{Eigenvalue splitting in the ``bulk''}

Let $\Omega\subset \R^2$ be a compact set and let $\set{\alpha(\kk)}_{\kk \in \Omega}$ be a smooth family of $m \times m$ unitary matrices. We want to show in this final Subsection that degeneracies of the eigenvalues of $\alpha$ can be lifted, and that one can find a family of unitary matrices arbitrarily close to $\alpha$ which has completely non-degenerate spectrum.

We begin with a formal definition of a cluster of eigenvalues.

\begin{definition}
Let $\epsilon>0$. We say that $c_n(\kk)\subset \sigma(\alpha(\kk))$ is an \emph{$\epsilon$-cluster of $n$ eigenvalues} of $\alpha(\kk)$ (counting multiplicities) if for every $\lambda,\mu\in c_n(\kk)$ we have $|\lambda-\mu|<\epsilon$. 
\end{definition}

The next Lemma shows that $\epsilon$-clusters are stable.

\begin{lemma}\label{lemma1}
If $\alpha(\kk_0)$ has an $\epsilon$-cluster $c_n(\kk_0)$, then $\alpha(\kk)$ has at least one $\epsilon$-cluster with $n$ eigenvalues whenever $\kk\in B_r(\kk_0)$ and $r$ is small enough. If $\alpha(\kk_0)$ does not have any $\epsilon$-clusters containing $n$ eigenvalues, then $\alpha(\kk)$ cannot have $\epsilon/2$-clusters with $n$ eigenvalues when $\kk\in B_r(\kk_0)$ and $r$ is small enough.
\end{lemma} 
\begin{proof} First, we assume that $\alpha(\kk_0)$ has an $\epsilon$-cluster $c_n(\kk_0)$. Consider a closed simple contour $\mathcal{C}$ in the complex plane which surrounds the $n$ eigenvalues in $c_n(\kk_0)$ and define the Riesz projection associated to them:
$$P(c_n(\kk_0))=\frac{1}{2\pi \iu}\int_\mathcal{C} (z\Id -\alpha(\kk_0))^{-1} \, \di z.$$ 
The Hausdorff distance $d\sub{H}(\sigma(\alpha(\kk)),\sigma(\alpha(\kk_0)))$ between the spectra of $\alpha(\kk)$ and $\alpha(\kk_0)$ is smooth, thus, if $r$ is small enough, $\mathcal{C}$ lies in the resolvent set of $\alpha(\kk)$ and we can define the spectral projection 
$$P(\kk)=\frac{1}{2\pi \iu}\int_\mathcal{C} (z\Id-\alpha(\kk))^{-1} \, \di z.$$ 
If $r$ is small enough, then $\norm{P(\kk)-P(c_n(\kk_0))}<1$, hence the range of $P(\kk)$ has dimension $n$ and corresponds to a cluster $c_n(\kk)$ which converges in the Hausdorff metric to $c_n(\kk_0)$. Finally, let $\lambda,\mu\in c_n(\kk)$. Then if $r$ is even smaller
$$|\lambda-\mu|\leq 2 d\sub{H}(c_n(\kk),c_n(\kk_0)) +\max_{a,b\in c_n(\kk_0)}|a-b|<\epsilon.$$

Second, we assume that $\alpha(\kk_0)$ has no $\epsilon$-clusters containing $n$ eigenvalues. In particular, each eigenvalue of $\alpha(\kk_0)$ has a degeneracy of order at most $n-1$. This implies that given any collection of $n$ eigenvalues (counting multiplicities) of $\alpha(\kk_0)$, at least two of them are at a distance larger or equal than $\epsilon$ from each other and are associated to mutually orthogonal projections. Now assume that there exists a sequence $\kk_\ell$ converging to $\kk_0$ such that $\alpha(\kk_\ell)$ has an $\epsilon/2$-cluster with $n$ eigenvalues. If $\ell$ is large enough, then near each eigenvalue $\lambda$ of $\alpha(\kk_0)$ with degeneracy $j<n$ there will be exactly $j$ eigenvalues of $\alpha(\kk_\ell)$ (counting multiplicities) which will converge to $\lambda$. Hence at least two of the $n$ eigenvalues from the $\epsilon/2$-cluster of $\alpha(\kk_\ell)$ must be close to different eigenvalues of $\alpha(\kk_0)$ if $\ell$ is large enough, thus at some point the distance between them must be larger or equal than $\epsilon/2$, a contradiction. 
\end{proof}

For any  $\epsilon>0$, we define $\Omega_{\epsilon,n}\subset \Omega$ to be the set of those $\kk$'s for which $\alpha(\kk)$ has at least one $\epsilon$-clusters of $n$ eigenvalues counting multiplicities.
 
\begin{lemma}\label{lemma2}
\begin{enumerate}
 \item $\Omega_{\epsilon,n}$ is an open set.
 \item Assume that $\alpha(\kk)$ can have eigenvalues with a degeneracy at most $n\leq m$, for all $\kk\in \Omega$.  Then there exists some $\epsilon_0$ and $A>0$ such that for any $\epsilon<\epsilon_0$ and for any $\epsilon$-cluster $c_n(\kk)$ we have
\begin{align}\label{hc1}
\left (\inf_{\kk\in \Omega_{\epsilon,n}}\inf_{\lambda\in c_n(\kk)}\inf_{\mu\in \sigma(\alpha(\kk))\setminus c_n(\kk)}|\lambda-\mu|\right ) \geq A >0. 
\end{align}
In other words, each such ``maximal'' cluster remains  at a distance at least $A$ from the rest of the spectrum, uniformly in $\kk\in \Omega_{\epsilon,n}$ and $\epsilon<\epsilon_0$.  
\end{enumerate}
\end{lemma} 
\begin{proof} 
The fact that $\Omega_{\epsilon,n}$ is open is a direct consequence of Lemma~\ref{lemma1}. 
 
Now let us assume that the infimum in \eqref{hc1} is zero. Then there exists a sequence $\epsilon_\ell\to 0$ and a sequence $\kk_\ell\in \Omega$ such that the matrix $\alpha(\kk_\ell)$ has an $\epsilon_\ell$-cluster $c_n(\kk_\ell)$, and there exist two eigenvalues $\lambda_\ell\in c_n(\kk_\ell)$ and $\mu_\ell\in \sigma(\alpha(\kk_\ell))\setminus c_n(\kk_\ell)$ such that $|\lambda_\ell-\mu_\ell|<1/\ell$. 
 
Passing to a subsequence, we may assume that $\kk_\ell$ converges to some $\kk_\infty\in \Omega$. Since $\alpha(\kk_\infty)$ can have eigenvalues with a  degeneracy at most $n$, there exists some $\delta$ small enough such that $\alpha(\kk_\infty)$ cannot have $\delta$-clusters containing $n+1$ eigenvalues (counting multiplicities). Since $\alpha(\kk_\ell)$ converges in norm to $\alpha(\kk_\infty)$, Lemma~\ref{lemma1} implies that $\alpha(\kk_\ell)$ cannot have  $\delta/2$-clusters which contain $n+1$ eigenvalues (counting multiplicities) if $\ell$ is large enough. But $c_n(\kk_\ell)\cup \{\mu_\ell\}$ is a $\delta/2$-cluster containing $n+1$ eigenvalues when $\epsilon_\ell+1/\ell<\delta/2$, a contradiction. 
\end{proof}
 
Let $S_{\epsilon,n}\subset \Omega$ be the set of all $\kk$'s where $\alpha(\kk)$ has at least one ``closed'' cluster $\overline{c_n(\kk)}$ of $n$ eigenvalues, {\sl i.e.} for any $\lambda,\mu\in \overline{c_n(\kk)}$ we have that $|\lambda-\mu|\leq \epsilon$. 

\begin{lemma}\label{lemma3}
The set $S_{\epsilon,n}$ is compact, and for every $\epsilon'>\epsilon$ we have 
$$\overline{\Omega_{\epsilon,n}}\subset S_{\epsilon,n}\subset \Omega_{\epsilon',n}.$$ 
\end{lemma} 
\begin{proof} 
We first prove that $S_{\epsilon,n}$ is closed. Let $\kk_0\in \Omega\setminus S_{\epsilon,n}$. Then given any group $c_n$ of $n$ eigenvalues of $\alpha(\kk_0)$, there exists at least one among them, call it $\lambda$ and assume that it has multiplicity $j<n$, such that the distance between $\lambda$ and the other $n-j$ eigenvalues is larger than $\epsilon$. Since $\alpha(\kk)$ is norm-smooth, the subcluster of $j$ eigenvalues evolving from $\lambda$ will still remain at a distance larger than $\epsilon$ from the other $n-j$ if $\norm{\kk-\kk_0}<r$ is small enough. Since the number of all possible $n$ combinations of eigenvalues is finite, we may find an $r>0$ which simultaneously works for all of them.  Thus $\Omega\setminus S_{\epsilon,n}$ is open. Since $S_{\epsilon,n}$ is contained in the compact set $\Omega$, it must be compact as well. 

We also have $S_{\epsilon,n}\subset \Omega_{\epsilon',n}$, while both $S_{\epsilon,n}$ and $\Omega\setminus \Omega_{\epsilon',n}$ are  compact. Thus the distance between $S_{\epsilon,n}$ and $\Omega\setminus \Omega_{\epsilon',n}$ is positive.
\end{proof}
 
Assume that $\alpha(\kk)$ can have eigenvalues of multiplicity at most $n$ when $\kk\in\Omega$ and let $\epsilon_1\leq \min\{\epsilon_0/10,A/10\}$, where $\epsilon_0$ and $A$ are as in \eqref{hc1}. We also assume that the maximal degeneracy $n$ is achieved in at least one point so that $S_{\epsilon_1,n}$ is never empty no matter how small $\epsilon_1$ is taken.

Given any $\kk_0\in S_{\epsilon_1,n}$ we know that there exist $\ell(\kk_0)\geq 1$ ``closed'' clusters of $n$ eigenvalues denoted by $\overline{c_n^{(j)}(\kk_0)}$, $1\leq j\leq  \ell(\kk_0)$, such that the distance between any two eigenvalues inside each cluster is less or equal than $\epsilon_1\leq A/10$, while the distance between any given cluster and its complementary part of the spectrum is at least $A$, uniformly in $\epsilon_1$. 
 
There are exactly $\ell(\kk_0)$ Riesz projections $P_j(\kk_0)$ of rank $n$, and each of them can be smoothly extended to a unique $P_j(\kk)$ near $\kk_0$. 

Let $d>0$ be the distance between $S_{\epsilon_1,n}$ and $\Omega\setminus \Omega_{2\epsilon_1,n}$. Define the open sets
\begin{equation} \label{hc2}
\begin{aligned}
O_1(\kk_0)&=B_{d/2}(\kk_0)\cap \{\kk\in \Omega_{2\epsilon_1,n}:\; \norm{P_j(\kk)-P_j(\kk_0)}<1/10,\; 1\leq j\leq \ell(\kk_0)\},\nonumber \\
O_2(\kk_0)&=B_{d}(\kk_0)\cap \{\kk\in \Omega_{3\epsilon_1,n}:\; \norm{P_j(\kk)-P_j(\kk_0)}< 1/5,\; 1\leq j\leq \ell(\kk_0)\}.
\end{aligned}
\end{equation}
The closure $\overline{ O_1(\kk_0)}$ is included in the set 
$$\overline{B_{d/2}(\kk_0)}\cap \{\kk\in \overline{\Omega_{2\epsilon_1,n}}:\; \norm{P_j(\kk)-P_j(\kk_0)}\leq 1/10,\; 1\leq j\leq \ell(\kk_0)\}.$$ 
 
Also, from Lemma \ref{lemma3}, we have that $\overline{\Omega_{2\epsilon_1,n}}\subset S_{2\epsilon_1,n}\subset \Omega_{3\epsilon_1,n}$, hence 
$$\overline{ O_1(\kk_0)}\subset O_2(\kk_0).$$
Because $S_{\epsilon_1,n}$ is compact, there exist a finite number $B$ of points $\kk_b$ such that 
\begin{align}\label{hc3}
S_{\epsilon_1,n}\subset \bigcup_{b=1}^B O_1(\kk_b).
\end{align}

Finally we introduce the set 
 \begin{align}\label{hc3'}
T_{\epsilon_1,n,L}:= \bigcup_{\ell(\kk_b)=L} \overline{O_1(\kk_b)}, \quad L:=\max_{1\leq b\leq B}\{\ell(\kk_b)\}\geq 1,.
\end{align}
This compact set contains all the points of $\Omega$ in which $\alpha$ might have exactly $L$ different $n$-fold degenerate eigenvalues, well separated among themselves and from the rest of the spectrum.

The next result allows to lift maximal degeneracies.

\begin{proposition}\label{prop1}
Assume that $\alpha(\kk)$ is smooth in $\kk$ and can have eigenvalues of multiplicity at most $n$, with $2\leq n\leq m$. We also assume that not more than one eigenvalue can be $n$-fold degenerate at a time. Then we can construct a family of smooth unitary maps $\alpha_s(\kk)$, $s>0$, such that 
$$\lim_{s\to 0}\sup_{\kk\in \Omega} \norm{\alpha_s(\kk)-\alpha(\kk)}=0$$ 
while $\alpha_s(\kk)$ can only have eigenvalue crossings of order at most $n-1$. 
\end{proposition}
\begin{proof} 
An example would be $m=5$ and $n=3$, when if a $3$-crossing occurs at some $\kk$, no other $3$-crossing is possible at the same $\kk$. 
 
Since no more than one $n$ cluster can occur if $\epsilon_1$ is small enough, then on each open set $O_1(\kk_b)$ in \eqref{hc3} we will eventually have $\ell(\kk_b)=1$, {\sl i.e.} there is exactly one $\epsilon_1$-cluster of maximal dimension $n$ near $\kk_b$, and we can put $L=1$ in \eqref{hc3'}.

We first want to prevent any $n$-fold degeneracies  which might happen in $\overline{O_1(\kk_1)}$. We know that there exists exactly one $3\epsilon_1$-cluster of $n$ eigenvalues corresponding to a smooth Riesz projection $P_1(\kk)$ if $\kk\in \overline{O_1(\kk_1)}$. 
We can intertwine $P_1(\kk)$ and $P_1(\kk_1)$ through a smooth Kato-Nagy unitary $U(\kk)$, such that $U(\kk)P_1(\kk)U^{-1}(\kk)=P_1(\kk_1)$. The eigenvalues belonging to the $n$ cluster coincide with the eigenvalues of the reduced operator
\begin{align}\label{hc4}
U(\kk)P_1(\kk)\alpha(\kk)U^{-1}(\kk)=P_1(\kk_1) U(\kk)\alpha(\kk)U^{-1}(\kk)P_1(\kk_1),\quad \kk\in O_2(\kk_1).
\end{align}
This operator is represented by an $n\times n$ unitary matrix $\gamma(\kk)$ when restricted to ${\rm Ran}(P_1(\kk_1))$. Denoting by $\{\Psi_i(\kk_1)\}_{i=1}^{n}$ some orthonormal basis in this range, we have
$$\gamma_{ij}(\kk):=\scal{\Psi_i(\kk_1)}{U(\kk)\alpha(\kk)U^{-1}(\kk)\Psi_j(\kk_1)} ,\quad 1\leq i,j\leq n.$$
The unitary matrix $\gamma(\kk)$ is almost diagonal because the distance between any two of its $n$ eigenvalues is at most $3\epsilon_1$. Define $z(\kk)={\rm Tr}(\gamma(\kk))/n$. Then we must have
$$\gamma(\kk)=z(\kk)\Id+\mathcal{O}(\epsilon_{1}),\quad |z(\kk)|=1+\mathcal{O}(\epsilon_{1}).$$
Taking the determinant on both sides we obtain
$$\det \gamma(\kk) = z(\kk)^n (1+\mathcal{O}(\epsilon_{1})).$$
If $\epsilon_1$ is chosen smaller than some numerical constant, then we can write 
$$\det \gamma(\kk) = \left (z(\kk) \eu^{\frac{1}{n}{\rm Ln}(1+\mathcal{O}(\epsilon_{1}))}\right )^n.$$

Let us define
$$\widetilde{\gamma}(\kk):=\left (z(\kk) \eu^{\frac{1}{n}{\rm Ln}(1+\mathcal{O}(\epsilon_{1}))}\right )^{-1}\gamma(\kk).$$
Then $\det \widetilde{\gamma}(\kk)=1$  and $\norm{\widetilde{\gamma}(\kk)-\Id}=\mathcal{O}(\epsilon_{1})$ for all $\kk\in O_2(\kk_1)$. An important consequence is that $\widetilde{\gamma}(\kk)$ is a smooth $SU(n)$ matrix, close in norm to the identity matrix. Moreover,
$$\widetilde{\gamma}(\kk)=\eu^{\iu h(\kk)}$$
where $h(\kk)$ is a traceless generator of $SU(n)$, self-adjoint, smooth, and $\norm{h(\kk)}= \mathcal{O}(\epsilon_{1})$.
Hence the $U(n)$ matrix $\gamma(\kk)$ can be expressed as
$${\gamma}(\kk)=z(\kk) \eu^{\frac{1}{n}{\rm Ln}(1+\mathcal{O}(\epsilon_{1}))}\eu^{\iu h(\kk)}.$$
The complete matrix $\alpha$ has an $n$-crossing in $O_2(\kk_1)$ \emph{if and only if} $\gamma$ is diagonal, and this is equivalent with $h(\kk)=0$.
The following estimate
\begin{align}\label{hc5}
\norm{h(\kk)}\geq C(\epsilon_1)>0,\quad  \kk\in O_2(\kk_1)\setminus T_{\epsilon_1,n,1}
\end{align}
tells us that at least one of the $n$ eigenvalues of our cluster must be at a positive distance from at least one other eigenvalue from the other $n-1$ if $\kk\in O_2(\kk_1)\setminus T_{\epsilon_1,n,1}$, hence $h$ must have at least one non-zero real eigenvalue, uniformly in $\kk$. 

Denote by $\Sigma_{1,2,3}$ the extensions of the usual Pauli matrices to $\mathfrak{su}(n)$. The decomposition of ${h}(\kk)$ with respect to the basis of the generators of $SU(n)$ has some components $F_1(\kk)\Sigma_1+F_2(\kk)\Sigma_2+F_3(\kk)\Sigma_3$, where each $F_j(\kk)={\rm Tr}({h}(\kk)\Sigma_j)/2$ is smooth on $O_2(\kk_1)$. The range of the map
$$O_2(\kk_1)\ni \kk\mapsto {\bf F}(\kk):=[F_1(\kk),F_2(\kk),F_3(\kk)]\in\R^3$$
has Lebesgue measure zero, hence the origin of $\R^3$ cannot be an interior point of this range. Thus given any $s_1>0$ we may find a vector ${\bf v}^{(s_1)}\in\R^3$ such that:
\[
\norm{{\bf v}^{(s_1)}}=s_1,\qquad \inf_{\kk\in \overline{O_1(\kk_1)}}\norm{{\bf F}(\kk)+{\bf v}^{(s_1)}}>0.
\]
Now take a smooth function $0\leq \chi\leq 1$ which equals $1$ on the compact set $\overline{O_1(\kk_1)}$ and has support in $O_2(\kk_1)$. Define the $U(n)$ matrix
$$\gamma^{(s_1)}(\kk):=z(\kk) \eu^{\frac{1}{n}{\rm Ln}(1+\mathcal{O}(\epsilon_{1}))}\exp\left\{\iu \left( {h}(\kk)+\chi(\kk) \,{\bf v}^{(s_1)}\cdot\boldsymbol{\Sigma} \right) \right\},\quad \kk\in O_2(\kk_1).$$
The perturbed generator ${h}(\kk)+\chi(\kk)\,{\bf v}^{(s_1)}\cdot\boldsymbol{\Sigma}$ is traceless, hence $\gamma^{(s_1)}(\kk)$ has $n$ identical eigenvalues in $O_2(\kk_1)$ if and only if ${h}(\kk)+\chi(\kk)\,{\bf v}^{(s_1)}\cdot\boldsymbol{\Sigma}=0$ for some $\kk\in O_2(\kk_1)$. 

If $s_1$ is sufficiently small compared to $\epsilon_{1}$,  then ${h}(\kk)+\chi(\kk)\,{\bf v}^{(s_1)}\cdot\boldsymbol{\Sigma}$ cannot become zero inside $O_2(\kk_1)\setminus T_{\epsilon_1,n,1}$  due to \eqref{hc5}.  Also, the same perturbed generator equals ${h}(\kk)+{\bf v}^{(s_1)}\cdot\boldsymbol{\Sigma}$ inside $\overline{O_1(\kk_1)}$ thus it is different from zero,  no matter how small $s_1>0$ is. We conclude that if $s_1$ is smaller than some critical value depending on $\epsilon_{1}$, the $U(n)$ matrix $\gamma^{(s_1)}(\kk)$ can have an $n$-crossing neither on $O_2(\kk_1)\setminus T_{\epsilon_1,n,1}$ nor on $\overline{O_1(\kk_1)}$. Moreover, 
$$\lim_{s_1\searrow 0}\sup_{\kk\in O_2(\kk_1)}\norm{\gamma^{(s_1)}(\kk)-\gamma(\kk)}=0.$$

Using the notation introduced in \eqref{hc4}, we define the $U(m)$ matrix $\alpha^{(s_1)}(\kk)$ for $\kk \in \Omega$ as follows:
\[
\alpha^{(s_1)}(\kk):= \begin{cases}
\alpha(\kk), &  \kk\in \Omega\setminus O_2(\kk_1) \\
P_1(\kk)^\perp\alpha(\kk)+ \sum_{i,j=1}^{m-1} \gamma_{ij}^{(s_1)}(\kk)\ket{U^*(\kk)\Psi_i(\kk_1)}\bra{U^*(\kk)\Psi_j(\kk_1)}, & \kk\in O_2(\kk_1)\end{cases}
\]
The matrix $\alpha^{(s_1)}(\kk)$ is smooth and converges in norm to $\alpha(\kk)$. Also, if $s_1$ is sufficiently small (compared to $\epsilon_{1}$), we know that $\alpha^{(s_1)}(\kk)$ can have $n$-crossings neither on the old $\Omega\setminus T_{\epsilon_1,n,1}$ 
nor on $\overline{O_1(\kk_1)}$. 

In the next step we will fix $s_1$ and perturb $\alpha^{(s_1)}(\kk)$ inside $O_2(\kk_2)$ so that the new perturbed matrix (denoted by $\alpha^{(s_1,s_2)}(\kk)$) has no $n$-crossings on the old $\Omega\setminus T_{\epsilon_1,n,1}$, on  $\overline{O_1(\kk_1)}$ \emph{and} on $\overline{O_1(\kk_2)}$, provided $s_2$ is small enough  with respect to \emph{both} $\epsilon_{1}$ 
and $s_1$. We stress the fact that the sets $O_2(\kk_j)$ and $O_1(\kk_j)$ (see \eqref{hc2}) remain unchanged and only depend on the initial $\alpha$. 

We observe that $\alpha^{(s_1)}(\kk)$ will continue to have exactly one well isolated $n$ cluster if $\kk\in O_2(\kk_b)$. Moreover, the Riesz projections $P_j^{(s_1)}(\kk)$ are close in norm to $P_j(\kk)$ when $s_1$ is small, hence the Kato-Nagy formula can again be used for $\kk\in O_2(\kk_2)$ in order to reduce the problem to an $n$-dimensional matrix as we did in \eqref{hc4}. 

We can perform the same type of local generator perturbation (depending on $s_2$), where now $\alpha$ is replaced with $\alpha^{(s_1)}$ and $O_2(\kk_1)$ with $O_2(\kk_2)$. The $s_2$-perturbation may affect the set $O_2(\kk_2)\cap\overline{O_1(\kk_1)}$, but if $s_2$ is small enough compared to $s_1$ and $\epsilon_{1}$ it cannot re-induce $n$-crossings inside $\overline{O_1(\kk_1)}$. On the other hand, we can make sure that the perturbed generator is never zero on $\overline{O_1(\kk_2)}$ no matter how small $s_2$ is. 

After $B$ steps we make sure that no $n$-crossings can take place in $T_{\epsilon_1,n,1}$ and we are done. 
\end{proof}

Finally, with the next result we are able to lift multiple simultaneous degeneracies.

\begin{proposition}\label{prop2}
Assume that $\alpha(\kk)$ is smooth and can have eigenvalues of multiplicity at most $n$, with $2\leq n\leq m$. We assume that exactly $L>1$ different  eigenvalues can be simultaneously $n$-fold degenerate at a given $\kk$. Then we can construct a family of smooth unitary maps $\alpha_s(\kk)$, $s>0$, such that 
$$\lim_{s\to 0}\sup_{\kk\in \Omega}\norm{\alpha_s(\kk)-\alpha(\kk)}=0$$ 
and $\alpha_s(\kk)$ can have at most $L-1$ different  eigenvalues which can be simultaneously $n$-fold degenerate at a given $\kk$. 
\end{proposition}

\begin{proof} 
An example would be $m=5$ and $n=2$, where we could have two different doubly degenerate eigenvalues at the same $\kk$, {\sl i.e.} $L=2$ in \eqref{hc3'} if $\epsilon_1$ is small enough. 

The strategy of the proof is quite close to the previous one. We need to slightly perturb the matrix $\alpha$ on the set $T_{\epsilon_1,n,L}$, see \eqref{hc3'}. We start with one of the sets $\overline{O_1(\kk_b)}$ for which we know that $\ell(\kk_b)=L$, {\sl i.e.} for every $\kk$ in this set, the spectrum of $\alpha(\kk)$ has precisely $L$ clusters of $n$ eigenvalues which are well separated from each other and from the rest of the spectrum. Up to a Kato-Nagy rotation, we can block-diagonalize $\alpha(\kk)$ and find $L$ unitary matrices $\gamma_j(\kk)\in U(n)$ which are almost diagonal and their spectrum coincide with the $j$-th cluster. 

Like in the previous case, we can write
$$\gamma_j(\kk)=\eu^{\iu \phi_j(\kk)} \eu^{\iu h_j(\kk)},\quad \kk\in O_2(\kk_b),\quad h_j\in \mathfrak{su}(n),$$
where $\phi_j$ is some smooth phase, and $\norm{h_j(\kk)}=\mathcal{O}(\epsilon_1)$. As before, we can slightly perturb each $h_j(\kk)$ such that the perturbed ones will never be zero on $\overline{O_1(\kk_b)}$, no matter how small the perturbation is.

The crucial difference compared to the case $L=1$ is that \eqref{hc5} becomes 
\[
\max_{1 \le j \le L} \norm{h_j(\kk)}\geq C(\epsilon_1)>0,\quad \text{for all } \kk\in O_2(\kk_1)\setminus T_{\epsilon_1,n,L}.
\]

In words, this says that outside $T_{\epsilon_1,n,L}$ we cannot have $L$ simultaneous $n$-crossings, only at most $L-1$. Hence one of the generators $h_j$ must have at least one eigenvalue away from zero. 

Now choosing as before $s_1$ small enough, the perturbation will not induce $L$ simultaneous crossings outside $T_{\epsilon_1,n,L}$, while lifting completely the $n$-crossings inside $\overline{O_1(\kk_b)}$. Overall, the conclusion is that we now can have at most $L-1$ simultaneous $n$-crossings inside $O_2(\kk_1)\setminus T_{\epsilon_1,n,L}$ and no crossings at all inside $\overline{O_1(\kk_b)}$. 

Then after a finite number of steps in which we use weaker and weaker perturbations each time, we exhaust $T_{\epsilon_1,n,L}$ and we are done.
\end{proof}

We are finally able to conclude with the desired completely non-degenerate generic form.
 
\begin{proposition}
Assume that $\alpha(\kk)\in U(m)$ is smooth on the compact set $\Omega$. Then, there exists a family of smooth unitary maps $\alpha_s(\kk)$, $s>0$, such that 
$$\lim_{s\to 0}\sup_{\kk\in \Omega}\norm{\alpha_s(\kk)-\alpha(\kk)}=0$$ 
and $\alpha_s(\kk)$ only has only non-degenerate eigenvalues. 
\end{proposition}
\begin{proof}
This general statement follows from successive applications of Propositions \ref{prop1} and \ref{prop2}. For each point at which $\alpha(\kk)$ has $L$ eigenvalues each of which is $n$-fold degenerate, the two Propositions allow to perturb the matrix locally in order to achieve that the approximant has either less $n$-fold degenerate eigenvalues ({\sl i.e.} one reduces $L$ to $L-1$), or at least one degeneracy is lifted ({\sl i.e.} one passes from $n$-fold to $(n-1)$-fold degenerate eigenvalues). After a finite number of steps, each degeneracy is lifted.

We give a detailed proof for $m=4$. Assume that $\alpha$ has crossings which involve all its eigenvalues, {\sl i.e.} $n=m=4$. Only one maximal cluster can occur in this case, so we can apply Proposition \ref{prop1} which gives us an approximation $\alpha'$  which can only have crossings of order $m-1=3$. 

If there exist points $\kk$ where $\alpha'$ can have crossings involving $n=3$ eigenvalues, then again there exists only one maximal cluster at a time, hence Proposition \ref{prop1} gives us another approximation $\alpha''$  which can only have crossings of order $n-1=2$.

If $\alpha''$ has crossings of order two, it might happen that there exist points $\kk$ where a simultaneous double degeneracy can occur ($n=2$, $2+2=4$). This corresponds to $L=2$ in Proposition \ref{prop2}. Applying this Proposition we obtain a new matrix $\alpha'''$ which can only have one ({\sl i.e.} $L-1$) crossing of $n=2$ eigenvalues at any given $\kk$. Finally, we apply 
Proposition \ref{prop1} in order to split the last degeneracy and we are done.
\end{proof}


\section{Outlook}\label{open-pro}

To conclude this paper, let us briefly state two interesting issues which have not been addressed before. 

In the fermionic case in $d=1$ ($D=2$) we gave in \cite{CorneanMonacoTeufel17} an explicit construction of an homotopy which interpolates between any two families of matching matrices having the same $\mathbb{Z}_2$ invariant, regardless of whether it vanishes or not. The proof relies on a certain special factorization of \emph{any} such family $\alpha$, obstructed or not, which allows one to combine the two similar families into an equivariantly null-homotopic one, which can then be deformed into the identity. The special factorization in question reads \cite[Lemma~5.1]{CorneanMonacoTeufel17}
\[ \alpha(k) = \eps^{-1} \, \gamma(-k)^t \, \eps \, \gamma(k), \quad k \in \R, \]
where $\set{\gamma(k)}_{k \in \R}$ is a continuous and $\Z$-periodic family of unitary matrices. This factorization is also useful to compute the GP-index of $\alpha$ as $\rueda(\alpha) = \deg(\det(\gamma)) \bmod 2$, namely as the reduction $\bmod \: 2$ of the winding number of the periodic map $\det \gamma \colon [-1/2,1/2] \to U(1)$ \cite[Prop.~5.1]{CorneanMonacoTeufel17}.

Assume that a similar factorization holds in $d=2$ ($D=3$), namely that
\begin{equation} \label{eqn:factor}
\alpha(\kk) = \eps^{-1} \, \gamma(-\kk)^t \, \eps \, \gamma(\kk), \quad \kk \in \R^2.
\end{equation}
Then the map $\gamma(k_1, \cdot)$ gives an homotopy between the family $\set{\gamma(k_1,0)}_{k_1 \in \R}$ and $\set{\gamma(k_1,1/2)}_{k_1 \in \R}$, implying in particular that $\rueda(\alpha(\cdot,0)) = \rueda(\alpha(\cdot,1/2))$. Similarly one would deduce that $\rueda(\alpha(0,\cdot)) = \rueda(\alpha(1/2,\cdot))$ by exchanging the roles of $k_1$ and $k_2$. We see then that the factorization \eqref{eqn:factor} can not hold in general, since the above constraints on the GP-indices only account for four out of the eight equivariant homotopy classes of families of fermionic matching matrices in $d=2$ provided by Corollary~\ref{crl:2Dhomotopies}\eqref{item:2DEquivariantF}. Indeed, the constraints above exclude those classes for which the restrictions on one horizontal and one vertical line have the same GP-index (say $\rueda(\alpha(\cdot,0)) = \rueda(\alpha(0,\cdot))$) but they disagree with the GP-indices on the ``opposite'' horizontal and vertical lines (say $\rueda(\alpha(\cdot,0)) \ne \rueda(\alpha(\cdot,1/2))$). Thus, even if Theorem~\ref{TFAE} gives an explicit homotopy between all equivariantly null-homotopic matrices (compare \eqref{ExplicitHomotopy}), constructing explicit homotopies between \emph{obstructed} families of fermionic matching matrices in $d=2$ remains an open problem, and we suspect that a completely new idea is needed. 

The second open issue is about the explicit construction of maximally \emph{exponentially} localized composite Wannier functions in $L^2(\mathbb{R}^D)$ corresponding to an isolated spectral band $\sigma_0$ (assuming that no obstructions are present). Let us be more precise. The integral kernel of the corresponding projection has an off-diagonal exponential decay $e^{-\delta_0 \norm{\x-\x'}}$ where $\delta_0$ equals the distance between $\sigma_0$ and the rest of the spectrum. In the Bloch picture, the corresponding projection $P(\mathbf{k})$ will have an analytic continuation to a complex open strip near $\mathbb{R}^D$ of width $\delta_0$. Our constructive algorithm can be adapted, as mentioned in the Introduction, to produce a Bloch frame which is periodic and \emph{real analytic}, that is, analytic in a smaller complex strip of width $\delta'<\delta_0$. The main question is whether one can push the analyticity of the periodic Bloch frame up to a strip of width $\delta_0$, which would provide composite Wannier functions having an exponential decay like $e^{-\delta\norm{\x}}$ with \emph{any} $\delta<\delta_0$. We note that the construction of Nenciu \cite{Nenciu91} in the rank-$1$ case does just that. The \emph{existence} of such exponentially localized composite Wannier functions can be proved by means of methods from the theory of functions of several complex variables (see \cite{Panati07, Kuchment09} and references therein), but to the best of our knowledge a \emph{constructive} proof is still missing in the literature, and constitutes a stimulating line of research for the future.



\bigskip \bigskip

{\footnotesize

\begin{tabular}{rl}
(H.D. Cornean) & \textsc{Department of Mathematical Sciences, Aalborg University} \\
 &  Fredrik Bajers Vej 7G, 9220 Aalborg, Denmark \\
 &  \textsl{E-mail address}: \href{mailto:cornean@math.aau.dk}{\texttt{cornean@math.aau.dk}} \\
 \\
(D. Monaco) & \textsc{Fachbereich Mathematik, Eberhard Karls Universit\"{a}t T\"{u}bingen} \\
 &  Auf der Morgenstelle 10, 72076 T\"{u}bingen, Germany \\
 &  \textsl{E-mail address}: \href{mailto:domenico.monaco@uni-tuebingen.de}{\texttt{domenico.monaco@uni-tuebingen.de}} \\
\end{tabular}

}

\end{document}